\documentclass[10pt,twocolumn,twoside]{IEEEtran}
\usepackage{amsmath}
\usepackage{amsfonts}
\usepackage{amssymb}
\usepackage{framed}
\usepackage{latexsym}
\usepackage{graphicx}
\usepackage[ruled,linesnumbered]{algorithm2e}
\SetKwRepeat{Do}{do}{while}%
\usepackage{epsfig,psfrag,color}
\usepackage{epstopdf}

\newtheorem{theorem}{\bf Theorem}[section]
\newtheorem{lemma}[theorem]{\bf Lemma}
\newtheorem{prop}[theorem]{\bf Proposition}

\newtheorem{note}[theorem]{\bf Note}

\newtheorem{remark}[theorem]{\bf Remark}
\newtheorem{question}[theorem]{\bf Problem}
\newtheorem{defi}[theorem]{\bf Definition}
\newenvironment{proof}{\noindent{\em Proof:}}{\quad \hfill$\Box$\vspace{2ex}}

\newcommand{\C}{\mathbb{C}}

\begin{document}

\title{Two-step PR-scheme for  recovering   signals in detectable  union of cones by magnitude measurements}

\author{Youfa~Li, Deguang~Han
\thanks{Copyright (c) 2018 IEEE}
\thanks{Youfa~Li. College of Mathematics and Information Science,
Guangxi University,  Naning, China.   Email: youfalee@hotmail.com}
\thanks{Deguang~Han. Department of Mathematics,
University of Central Florida, Orlando, FL 32816. Email:deguang.han@ucf.edu}

\thanks{Youfa Li is partially supported by Natural Science Foundation of China (Nos: 61561006, 11501132), Natural Science Foundation of Guangxi (No: 2016GXNSFAA380049)
and the talent project of  Education Department of Guangxi Government  for Young-Middle-Aged backbone teachers. Deguang Han  is partially supported by the NSF grant DMS-1403400 and DMS-1712602.}}

\date{}

\maketitle

\begin{abstract}
Motivated by the  research on sampling problems for a union of subspaces (UoS), we investigate in this paper the phase-retrieval problem for the signals that are residing in a union of (finitely generated) cones (UoC for short) in $\mathbb{R}^{n}$. We propose a two-step PR-scheme: $\hbox{PR}=\hbox{detection}+\hbox{recovery}$. We first establish a sufficient and necessary condition
for the detectability  of a UoC, and then design a detection  algorithm  that allows us to determine the cone where the target signal  is residing. The phase-retrieval will be then performed within the detected cone, which can be achieved by using at most $\Gamma$-number of measurements and  with very  low complexity, where $\Gamma (\leq n)$ is the maximum  of the ranks of the generators for  the UoC.  Numerical experiments are provided to demonstrate  the efficiency of our approach, and  to exhibit  comparisons  with some existing phase-retrieval methods.
\end{abstract}

\bigskip
\bigskip


\section{Introduction}\label{section1}
Phase-retrieval is a nonlinear  problem that seeks to
recover  a signal   $\textbf{x}$, up to a global phase ambiguity,  from   the magnitudes of  its  linear  measurements \begin{align}\notag  b_{i}:=|\langle \mathbf{x}, a_{i}\rangle|, i=1, \ldots, m.\end{align}
Phase-retrieval has been
widely applied in many applications such as  X-ray crystallography (\cite{Drenth}), quantum tomography (\cite{Heinosaarri}),  audio  processing (\cite{Becchetti})
and frame theory (\cite{Ba1,Ba3,Ba4,Sun3}).

Besides phase-retrieval,
the sampling theory for  a  union of subspaces   (UoS for short) is another important sampling problem  (c.f. \cite{LuDo,UOS1,Xampling}).  In signal processing, while traditionally we work on signals in a single linear space or subspace, there are practical demands requiring us to deal with the signals that lie in a UoS. A typical example is the sparse signal recovering or compressed sensing (c.f. \cite{Donoho}) where the signals are sitting in the finite union of  (small dimensional) subspaces.  M. Mishali,  Y. Eldar and  A. Elron \cite{Xampling}
established Xampling for recovering signals in the  UoS of $L^{2}(\mathbb{R})$. By  Xampling,  the target  subspace where the  signal sits   is detected before recovery.
As mention in \cite{Xampling}, the detection can considerably reduce the computational complexity  and measurement cost (the sampling rate).  Note that the phase information of  the measurements in \cite{Xampling}
is assumed   known.
Motivated by \cite{Xampling}  we will study the phase-retrieval problem  for the union of cones (UoC for short).  In order to introduce  the main problems and discuss our main contributions, we  need to recall and establish some notations and definitions.

\subsection {Notations and definitions}
We use  boldface letters   to denote column  vectors, e.g., $\textbf{x},$ calligraphic
and upper-case  letters to denote matrices (operator), e.g., $\mathcal{X}$, and underlined  letters
to denote a random variable, e.g., $\underline{\epsilon}$. For a matrix $\mathcal{X}$, its Hermitian transpose and  transpose
are  denoted by $\mathcal{X}^{*}$ and  $\mathcal{X}^{T}$, respectively.
For a linear operator $\mathfrak{P}$ from a vector space  $H_{1}$ to another  vector space $H_{2}$, we denote by $\mathcal{R}(\mathfrak{P})$ and $\mathcal{N}(\mathfrak{P})$
the range and null spaces of $\mathfrak{P}$, respectively. Moreover, for a set $S\subset \mathcal{R}(\mathfrak{P})$, denote by $\hbox{invim}(S)$
the inverse image of $S.$  For  a   vector $\textbf{x} \in \mathbb{R}^{m}$,  $\textbf{x}\succ 0 \ (\textbf{x}\prec0)$ implies every coordinate of $\textbf{x}$ is strictly larger (smaller) than $0$. Denote  $\mathbb{R}^{+,m}:=\{\textbf{x}\in \mathbb{R}^{m}: \textbf{x}\succ 0\}$. The standard orthornormal basis for $\mathbb{R}^{n}$
is denoted by $\{\textbf{e}_{1}, \ldots, \textbf{e}_{n}\}$.

For  a  matrix  $\mathcal{X}=[\textbf{x}_{1},\ldots, \textbf{x}_{m}]\in \mathbb{R}^{n\times m}$, we denote by $\textbf{\hbox{cone}}(\mathcal{X})$ the  cone generated from  its column vectors, i.e.,
\begin{align}\label{tukjian} \textbf{\hbox{cone}}(\mathcal{X}):=\{\theta_{1}\textbf{x}_{1}+\cdots+\theta_{m}\textbf{x}_{m}|  \theta_{i}\geq0, i=1, \ldots, m\}.\end{align}
A finite set  of vectors $\{\textbf{x}_{1},\ldots, \textbf{x}_{m} \}\subseteq \mathbb{R}^{n}$  is called a {\it
frame} for $\mathbb{R}^{n}$ if there exist two constants $0< C_{1}\leq C_{2}$ such that
\begin{align}\label{framedimention} C_{1}\|\textbf{z}\|_{2}^2 \leq \sum_{i=1}^{m} |\langle \textbf{z},\textbf{x}_i  \rangle|^2 \leq C_{2}\|\textbf{z}\|_{2}^2 \end{align}
holds for every $\textbf{z}\in \mathbb{R}^{n}$. Equivalently, a finite set is a frame for $\mathbb{R}^n$ if and only if it is a spanning set of $\mathbb{R}^n$.
The cone in \eqref{tukjian} is called a  {\it frame cone} if $\{\textbf{x}_{1},\ldots, \textbf{x}_{m}\}$ is a frame for $\mathbb{R}^{n}$.

Based on the above denotations, in what follows   we propose the   definition of a detectable UoC.

\begin{defi}\label{detectable} We say that    $\bigcup^{L}_{k=1} \textbf{\hbox{cone}}(\mathcal{X}_{k})$  is detectable,  if there exists a so-called
detector  $\mathcal{G}:=[\textbf{g}_{1}, \ldots, \textbf{g}_{\kappa}]\in\mathbb{R}^{n\times \kappa}$ such that for any nonzero target signal  $\textbf{z}\in\bigcup^{L}_{k=1} \textbf{\hbox{cone}}(\mathcal{X}_{l})$, the unique  index $l$  can be determined by  the detection  measurements $\{|\langle \textbf{g}_{1}, \textbf{z}\rangle|, \ldots, |\langle  \textbf{g}_{\kappa}, \textbf{z}\rangle|\}$ so that   $\textbf{z}\in\textbf{\hbox{cone}}(\mathcal{X}_{l})$.
\end{defi}

\subsection{Our goals, schemes and problems  in the present paper}\label{mudi}
The   union of cones (UoC) is an important type of  subset  of $\mathbb{R}^{n}$ which has been widely considered in many areas of research such as operations research (c.f. \cite{cone4,cone2,cone5}), signal processing (c.f. \cite{Donoho,cone6,cone3,P. Yu}), and representation theory (c.f. \cite{cone1}). Incidentally, since a linear space is a special type of cone (e.g. $\mathbb{R}^{2}$ can be regarded as the cone generated   from $\{\textbf{e}_{1}, -\textbf{e}_{1}, \textbf{e}_{2}, -\textbf{e}_{2}\}$), \emph{a union of linear spaces can be regarded as a UoC}.

It is well-known that the \emph{computational complexity}  and the \emph{amount of measurements}   are   two important considerations  for  the performance of  any  phase-retrieval method (c.f. \cite{Ba1,Ba9,WY}). The goal  of  this paper is to establish a phase retrieval  method in a UoC having low computational complexity and  requiring   very  few measurements.
Motivated by \cite{Xampling}, this goal will be  achieved by establishing the following two-step PR-scheme:
\begin{align}\label{step} \hbox{phase retrieval}=\hbox{detection}+\hbox{recovery}.\end{align}
Naturally, we need to address the issues in the following problem:
\begin{question}\label{wenti}
Under what conditions,  is a UoC detectable, namely,  the target cone  can be detected by  magnitude measurements?
Is it possible   to   utilize the   detectability   to reduce the amount of phase-retrievable measurement vectors  and  computational complexity (e.g. it can be $O(n)$ or  the FFT complexity $O(n\log n)$)?
\end{question}



\subsection {Existing results and our contributions}
In what follows we introduce our main contribution in this paper from the aspects  of cost of measurements and computational complexity (computational cost).

\subsubsection {Cost of measurements}
For the detection, we   establish the necessary and sufficient condition on the  detectability of   $\bigcup^{L}_{k=1} \textbf{\hbox{cone}}(\mathcal{X}_{k})$. Based on the condition we design a detector  $\mathcal{G}:=[\textbf{g}_{1}, \ldots, \textbf{g}_{L-1}]\in\mathbb{R}^{n\times (L-1)}$
and the detection  algorithm to detect the target cone. The detection can be achieved by using  only  $(L-1)$-number of  measurements.
Once the detection is completed, we then perform the phase-retrieval on  a single cone. As will be discussed  in Remark \ref{rem2.1}, there are at least $L-1$ cones, e.g. $\textbf{\hbox{cone}}(\mathcal{X}_{i}), i=1, \ldots, L-1$ in the detectable   $\bigcup^{L}_{k=1} \textbf{\hbox{cone}}(\mathcal{X}_{k})$ which   satisfy  the following overlap property
  \begin{align}\label{diyixx}
 \mathcal{R}(\mathcal{X}_{i}^{T})\cap \mathbb{R}^{+,m_{i}}\neq\emptyset.
 \end{align}
For the target  cone   $\textbf{\hbox{cone}}(\mathcal{X}_{i})$ in \eqref{diyixx},  we will  design $\hbox{rank}(\mathcal{X}_{i})$-number
of measurement vectors for the phase retrieval. Our contribution  on the amount of measurements is that  if all the $L$ cones satisfy \eqref{diyixx}, then ($L-1+\Gamma$)-number of
measurement vectors are sufficient for the two-step PR-scheme \eqref{step},   where $\Gamma=\max_{k}\{\hbox{rank}(\textbf{\hbox{cone}}(\mathcal{X}_{k}))\}$.


We emphasize two features of this approach.  (i) By the {\it complement property} for phase retrievable frames  (c.f. \cite{Ba1,Ba3,Ba4}), we know that any  phase-retrieval  method that applies to  the signals  in $\mathbb{R}^n$ requires at least  $2n-1$ measurement vectors.
Obviously, for many detectable UoCs,  the scheme \eqref{step}
requires  much fewer than $2n-1$ measurements $(L-1+\Gamma<2n-1).$ (ii) It is well known that the amount of measurement vectors can be significantly reduced for  sparse signals (e.g. \cite{Hand,Qiyucompressed}). In our case, $\Gamma$ being  small does not necessarily imply that the signals in $\bigcup^{L}_{k=1} \textbf{\hbox{cone}}(\mathcal{X}_{k})$ are  sparse.  So the reduction strategy  for the  amount of  measurements by scheme  \eqref{step} is different from the treatment of sparse signals.

\subsubsection {Computational complexity}

By using the i.i.d Gaussian measurement  vectors,
 E. Candes, Y. Eldar, T. Strohmer and V. Voroninski \cite{E.J. Candes}
proposed  the well-known PhaseLift  method to recover   $\textbf{z}$ in $\mathbb{R}^{n}$ (or $\C^{n}$).
Since then, based on the random measurements,  many  other efficient  phase-retrieval methods  such as
 Wirtinger Flow \cite{Wirting}, Alternating Minimization \cite{Netrapalli}, PhaseCut \cite{Mallat} and BlockPR \cite{WY}   have been proposed.  Among the above methods, the BlockPR, which holds for flat signals, has the lowest computational  complexity $O(n\log^{4}n)$.The signals  in  a  cone   may not be necessarily flat, and so they  do  not necessarily satisfy the condition required for the  BlockPR method. However,  by exploiting the structure of the detectable   UoC, the goal of significantly reducing  the computational complexity can  also be achieved. 
 Our  Algorithm \ref{separating} for detection costs $O(Ln)$-number of operations.
Theorem \ref{computational}  shows that if the target signal  lies in the cone $\textbf{\hbox{cone}}(\mathcal{X}_{i})$ satisfying  \eqref{diyixx}, then after detection, the phase-retrieval of the target signal  can be completed by $O(\gamma\log \gamma)$-operations, where $\gamma=\hbox{rank}(\mathcal{X}_{i})$.  Our contribution on the computational complexity is that  the proposed  phase-retrieval scheme \eqref{step} for a detectable union of $L$-cones all  satisfying \eqref{diyixx} has the computational  complexity $O(\Gamma\log \Gamma)+O(Ln)$, which
can be $O(n)$ or  $O(n\log n)$ for many cases of  $L$ and $\Gamma$.


\section{Two-step scheme  for recovering signals in detectable  union of cones}\label{section2}
Our PR-scheme \eqref{step} consists of detection and  recovery. In Subsection \ref{detection} we establish the sufficient and necessary condition
for the detectability of a UoC. The algorithm for this detection is presented in Algorithm 1.  Following this we discuss in Subsection \ref{sdfd} (Remark \ref{rem2.1})
 the cone structure derived from the above condition that is crucial to help achieve our goal. We also found that a union of linear subspaces (or spaces) is not detectable (Remark \ref{lindar}). The main results on the recovery will be presented in  Subsection \ref{fastrecovv}.

\subsection{Detection}\label{detection}
This subsection aims at establishing the sufficient and necessary condition for the detectability of a UoC, and presenting a detection   algorithm for the target cone.
\begin{theorem}\label{twocone}
A UoC $\bigcup^{L}_{k=1} \textbf{\hbox{cone}}(\mathcal{X}_{k})$, where  $\mathcal{X}_{k}=[\textbf{x}_{k,1},\ldots, \textbf{x}_{k,m_{k}}]\subseteq \mathbb{R}^{n\times m_{k}}$,  is detectable  if and only if for every $k(\geq2)$ we have either
\begin{align}\label{separatingPRgeneral}\begin{array}{lllllllllllllllll} \hbox{invim}(\mathcal{R}(\mathcal{X}_{l}^{T})\cap \mathbb{R}^{+,m_{l}})\cap \mathcal{N}(X^{T}_{k})\neq\emptyset\\
\hbox{or}\\
\hbox{invim}(\mathcal{R}(\mathcal{X}_{k}^{T})\cap \mathbb{R}^{+,m_{k}})\cap \mathcal{N}(X^{T}_{l})\neq\emptyset,
 \end{array}
\end{align}
where $ l=1, \ldots, k-1.$
\end{theorem}
\begin{proof}
The proof is given in Subsection \ref{proofttwocone}.
\end{proof}
Suppose that $\textbf{z}\in \bigcup^{L}_{k=1} \textbf{\hbox{cone}}(\mathcal{X}_{k})$.
If, for example, the first equation in  \eqref{separatingPRgeneral} holds, pick $\textsf{g}\in\hbox{invim}(\mathcal{R}(\mathcal{X}_{l}^{T})\cap \mathbb{R}^{+,m_{l}})\cap \mathcal{N}(X^{T}_{k})$,
then we determine that $\textbf{z}\notin \textbf{\hbox{cone}}(\mathcal{X}_{l})$ when $|\langle \textbf{z}, \textsf{g}\rangle|=0$,
and $\textbf{z}\notin \textbf{\hbox{cone}}(\mathcal{X}_{k})$ when $|\langle \textbf{z}, \textsf{g}\rangle|>0$.
It is easy to see that based on  \eqref{separatingPRgeneral}, we can use the $L-1$ exclusions similar to the above
to detect the target cone. The detection can be completed   by  using Algorithm \ref{separating}.

\begin{algorithm}\label{separating}
    \caption{Detection of the source of any $\textbf{z}\in\bigcup^{L}_{k=1} \textbf{\hbox{cone}}(\mathcal{X}_{k})$.}
 \KwIn{$[\mathcal{X}_{1}, \ldots, \mathcal{X}_{L}]$. }
 $s\gets 1;$

\For{$k=1:(L-1)$}
{\If{$\hbox{invim}(\mathcal{R}(\mathcal{X}_{s}^{T})\cap \mathbb{R}^{+,m_{l}})\cap \mathcal{N}(X^{T}_{k+1})\neq\emptyset$}
{Pick $\textsf{g}\in\hbox{invim}(\mathcal{R}(\mathcal{X}_{s}^{T})\cap \mathbb{R}^{+,m_{s}})\cap \mathcal{N}(X^{T}_{k+1})$;

\If{$|\langle \textbf{z}, \textsf{g}\rangle|=0$}
{$s\gets k+1;$}
}
\If{$\hbox{invim}(\mathcal{R}(\mathcal{X}_{k+1}^{T})\cap \mathbb{R}^{+,m_{2}})\cap \mathcal{N}(X^{T}_{s})\neq\emptyset$}
{Pick $\textsf{g}\in\hbox{invim}(\mathcal{R}(\mathcal{X}_{k+1}^{T})\cap \mathbb{R}^{+,m_{k+1}})\cap \mathcal{N}(X^{T}_{s})$;

\If{$|\langle \textbf{z}, \textsf{g}\rangle|>0$}
{$s\gets k+1;$}
}
}
\KwOut{$\textbf{z}\in \textbf{\hbox{cone}}(\mathcal{X}_{s})$.}
\end{algorithm}
\subsection{Remarks on the detectable union of cones}\label{sdfd}

\begin{remark}\label{rem2.1} (i) By  Algorithm \ref{separating},  the source of  any $f\in \bigcup^{L}_{k=1} \textbf{\hbox{cone}}(\mathcal{X}_{k})$ can be detected
 through  $L-1$ exclusions if   condition in  \eqref{separatingPRgeneral}   is satisfied.
Only one measurement vector is required for every exclusion. Therefore we need  $(L-1)$-number of  measurement vectors   for the target   cone detection.
Moreover the detection requires $O(Ln)$-number of operations.
(ii) The condition  \eqref{separatingPRgeneral} implies that the overlap property   \begin{align} \label{ufff}\mathcal{R}(\mathcal{X}^{T})\cap \mathbb{R}^{+,m}\neq\emptyset\end{align} holds for at leat $L-1$ number of cones.
\end{remark}

As mentioned in Section \ref{section1}, a linear space (subspace) is a special type of cone. An interesting problem is:  \emph{can the union of linear  spaces (subspaces)
be detectable}?  The following remark tells us that a detectable UoC has at most one of the cones that is a linear subspace (This can be easily proved by Remark \ref{rem2.1} (ii) and  the fact that a cone satisfying \eqref{ufff} is not  a  linear subspace). That is for any signal in the union of linear  spaces (subspaces), the target cone where the signal is residing can not be detected by magnitude measurements.

\begin{remark}\label{lindar}
Suppose that the UoC  $\bigcup^{L}_{k=1} \textbf{\hbox{cone}}(\mathcal{X}_{k})$ is detectable. Consequently, there exist at least $L-1$
cones satisfying the overlap property \eqref{ufff}, and none of the $L-1$ cones is a  linear space (subspace). If there exists a  linear space (subspace)
among the $L$ cones, then it is the unique one and does not have the overlap property \eqref{ufff}. In other words, \emph{a union of linear  spaces (subspaces)
is not detectable}, and it does not satisfy the requirement for the proposed approach.
\end{remark}

 In the following remark we discuss how to check \eqref{separatingPRgeneral} and \eqref{ufff}.

\begin{remark}
The condition \eqref{ufff} is equivalent to  that the system of linearly  inequalities
\begin{align}\label{budengshi}
\mathcal{X}^{T}\textbf{x}\succ 0
\end{align}
has a solution.
There exist  many   methods (e.g. in \cite{anual1,anual2,anual3}) in the liturature that can be used  to determine whether the \eqref{budengshi} has a solution.
The condition $\hbox{invim}(\mathcal{R}(\mathcal{X}_{l}^{T})\cap \mathbb{R}^{+,m_{l}})\cap \mathcal{N}(X^{T}_{k})\neq\emptyset$ in  \eqref{separatingPRgeneral}
is equivalent to  that the optimum of the  following quadratic programming problem
\begin{align}\label{mmyp012}
\left\{\begin{array}{lllllllllllllllll}
\min ||\mathcal{X}^{T}_{k}\textbf{x}||^{2}_{2} \\
\hbox{s.t.}\ \mathcal{X}_{l}^{T}\textbf{x}\succ0,
\end{array}\right.
\end{align}
is zero.

\end{remark}

\subsection{Recovery}\label{fastrecovv}
After the detection  by the procedures outlined in  Algorithm \ref{separating}, we  can  detect   the cone  that contains  the target signal.
What  left is to perform  phase retrieval   on the target cone  but not on the entire set UoC. As discussed in Section \ref{section1}, applying some of the existing  methods to a finitely generated cone  is either too expensive in terms of computational  complexity and measurements or  not even applicable due to the restriction of the methods. For example, the recently proposed fast method  BlockPR by  M. A. Iwen, A. Viswanathan, and Y. Wang \cite{WY} applies to flat vectors, but does not necessarily applies to vectors in a cone.  In this subsection  we establish a fast  PR method for  the cone in a detectable UoC with relatively fewer measurements and low computational complexity. The main results are outlined in  Theorem \ref{computational}, Theorem  \ref{sfdd} and Proposition \ref{complex}.

\begin{theorem}\label{computational}   Let $\textbf{\hbox{cone}}(\mathcal{X})$ be a  cone with $\mathcal{X}=[\textbf{x}_{1},\ldots, \textbf{x}_{m}]\in \mathbb{R}^{n\times m}$ such  that the overlap property  \eqref{ufff} holds.
Then there   exist $\gamma$-vectors $\{\textsf{f}_{k}\}^{\gamma}_{k=1}$  such that $\{|\langle \textbf{z}, \textsf{f}_{k}\rangle |\}^{\gamma}_{k=1}$ determines $\textbf{z}$ (up to a unimodular scalar) for any $\textbf{z}\in \textbf{\hbox{cone}}(\mathcal{X})$, where $\gamma=\hbox{rank}(\mathcal{X})$.
Moreover,  $\{\textsf{f}_{k}\}^{\gamma}_{k=1}$  can be designed in such a way that the recovery of  $\textbf{z}$ requires only  $O(\gamma\log \gamma)$-number of operations, i.e., the computational cost is FFT-time.
\end{theorem}
\begin{proof}
The proof is given in Section \ref{ssss}.
\end{proof}

Theorem \ref{computational} implies that the  property \eqref{ufff} is crucial for reducing the amount of measurements and computational
complexity for  the PR in a cone. By Remark \ref{rem2.1} (ii) there are at least $L-1$ cones in the detectable UoC $\bigcup^{L}_{k=1} \textbf{\hbox{cone}}(\mathcal{X}_{k})$ which satisfy  \eqref{ufff}.
We have the following result for the case  when all the $L$ cones in $\bigcup^{L}_{k=1} \textbf{\hbox{cone}}(\mathcal{X}_{k})$ satisfy \eqref{ufff}.
\begin{theorem}\label{sfdd}
Suppose that $\bigcup^{L}_{k=1} \textbf{\hbox{cone}}(\mathcal{X}_{k})$ is detectable, and all the $L$ cones satisfy the overlap property \eqref{ufff}.
Then, by using the two-step PR-scheme \eqref{step}, any target signal in the UoC can be determined by at most $L-1+\Gamma$ magnitude measurements, where $\Gamma=\max_{k}\{\hbox{rank}(\mathcal{X}_{k})\}$.
Moreover, our scheme costs at most  $O(Ln)+O(\Gamma\log \Gamma)$-number of  operations.
\end{theorem}

\begin{proof}
By Remark \ref{rem2.1}(i), the detection  strategy in  Algorithm \ref{separating} needs  $L-1$ magnitude measurements.
After the detection step, the phase-retrieval is performed  on  the target cone. Since all the cones satisfy the overlap property \eqref{ufff}, by Theorem \ref{computational}
the phase-retrieval on  the target cone needs at most $\Gamma$ magnitude measurements. Then $L-1+\Gamma$ measurements are  sufficient  for  the  two-step PR-scheme. The rest of the proof can be concluded by Remark \ref{rem2.1}(i) and Theorem \ref{computational}.
\end{proof}


The following proposition  states that for many cases of $L$ and $\Gamma$, the scheme \eqref{step} requires  very few measurements and has very  low computational complexity.

\begin{prop}\label{complex} (i)
The smaller $L+\Gamma$, the fewer measurements we need for our PR scheme \eqref{step}. In particular, when
$L+\Gamma<2n$ we can use less then $2n-1$ measurements (the critical amount  related to complement property) to complete  our PR scheme.

 (ii) As for the computational complexity, if $\Gamma\log \Gamma\lesssim n$ and $L$ is a constant independent of $n$, then our scheme can be
performed by $O(n)$-number of  operations. If $\Gamma\approx n$, then our scheme can be
done by $O(n\log n)$-number of  operations, the FFT time.\hfill $\blacksquare$
\end{prop}

\begin{remark}\label{necessary} (i)
Suppose that   $\mathcal{X}_{k}=[\textbf{x}_{k,1},\ldots, \textbf{x}_{k,m_{k}}]$ satisfies  \eqref{ufff}, i.e.,  $\mathcal{R}(\mathcal{X}_{k}^{T})\cap \mathbb{R}^{+,m_{k}}\neq\emptyset$. Then $\textbf{\hbox{cone}}(\mathcal{X}_{k})$ never contains  the unit  ball of  $\mathbb{R}^n$. (ii) Suppose that  $\bigcup^{L}_{k=1} \textbf{\hbox{cone}}(\mathcal{X}_{k})$ is detectable and all the $L$ cone generators satisfy \eqref{ufff}. Then $\bigcup^{L}_{k=1} \textbf{\hbox{cone}}(\mathcal{X}_{k})$ does not contain  the unit  ball of  $\mathbb{R}^n$ if  $L+\Gamma< 2n$.
\end{remark}

\begin{proof}
We first prove Part (i).  By Theorem \ref{computational}, there exist $n$ phase retrievable    vectors  for $\textbf{\hbox{cone}}(\mathcal{X}_{k})$. If the   unit  ball
$ \textbf{B}\subseteq \textbf{\hbox{cone}}(\mathcal{X}_{k})$, then the  $n$  vectors above  can also do PR for $ \textbf{B}$ and for $\mathbb{R}^n$.
 By the   complement property in \cite{Ba1}, however,  it requires at least $2n-1$ vectors to do PR for
$\mathbb{R}^n$ and also   for the unit ball. This is a contradiction, and the proof is concluded. Part (ii) can be proved similarly  by Theorem \ref{sfdd}
and the complement property.
\end{proof}

\section{Proof of Theorem \ref{computational}, algorithm for the phase-retrievable measurement  vectors, and the recovery formula}\label{ssss}

Before proving Theorem \ref{computational} and presenting an  algorithm for  $\{\textsf{f}_{k}\}^{\gamma}_{k=1}$ therein, we need some preparations.
Recall that  $\textbf{\hbox{cone}}(\mathcal{X})$ in Theorem \ref{computational} may not be   a frame  cone. However,
 the cone in Lemma \ref{gushi}
or  Lemma \ref{tjjj} will be required  to be a frame-type. In order to avoid notation confusions, we  will use   $\textbf{\hbox{cone}}(\mathcal{Y})$ instead of $\textbf{\hbox{cone}}(\mathcal{X})$ before  we present the proof of Theorem \ref{computational}, where $\mathcal{Y}\in\mathbb{R}^{n\times m}$.

Suppose that the column vectors  of $\mathcal{Y}=[\textbf{y}_{1},\ldots, \textbf{y}_{m}]$ constitute  a frame of $\mathbb{R}^{n}$, and the overlap property  \eqref{ufff} holds for $\mathcal{Y}.$
For any $\textbf{z}:=(z_{1},\ldots, z_{m})^{T}\in \mathcal{R}(\mathcal{Y}^{T})\cap\mathbb{R}^{+,m},$
it is easy to check by the frame property \eqref{framedimention}  that
\begin{align}\label{P1} \textsf{p}:=(\mathcal{Y}\mathcal{Y}^{T})^{-1}\mathcal{Y}\textbf{z}\end{align}
is the unique solution  to the following equation with respect to the variable  $\textbf{x}\in \mathbb{R}^{n}$,
\begin{align}\label{dyibu}
(\langle \textbf{x}, \textbf{y}_{1}\rangle, \langle \textbf{x}, \textbf{y}_{2}\rangle, \ldots, \langle \textbf{x}, \textbf{y}_{m}\rangle)^{T} =\textbf{z}.
\end{align}
Since  the measurements $\langle \textsf{p}, \textbf{y}_{1}\rangle, \ldots, \langle \textsf{p}, \textbf{y}_{m}\rangle$  are all positive,
 we will call   $\textsf{p}$  an \emph{anchor vector}.

\subsection{Two  auxiliary lemmas and design of special  anchor vector}

\begin{lemma}\label{gushi}   Let $\mathcal{Y}=[\textbf{y}_{1},\ldots, \textbf{y}_{m}]\in \mathbb{R}^{n\times m}$ and  $\textbf{\hbox{cone}}(\mathcal{Y})$ be a frame cone of $ \mathbb{R}^{n}$ such that \eqref{ufff} holds, i.e.,
$\mathcal{R}(\mathcal{Y}^{T})\cap \mathbb{R}^{+,m}\neq\emptyset$.
Then $\mathcal{R}(\mathcal{Y}^{T})\cap \mathbb{R}^{+,m}$ contains  $n$-linearly independent vectors.
\end{lemma}
\begin{proof}
If $m=n$, then the $n$-column vectors  of $\mathcal{Y}$ are a  basis  of $\mathbb{R}^{n}$.
Naturally, in this case,  $\mathcal{R}(\mathcal{Y}^{T})= \mathbb{R}^{n}$  and  the result holds.
We  next  prove the lemma  for the case of  $m>n.$ Without losing generality, we  assume that the first $n$-column vectors  $\{\textbf{y}_{1},\ldots, \textbf{y}_{n}\}$ of $\mathcal{Y}$ form a basis of $\mathbb{R}^{n}$.
Let  $\textbf{z}_{1}:=(z_{1,1}, \ldots, z_{m,1})^{T}\in \mathcal{R}(\mathcal{Y}^{T})\cap \mathbb{R}^{+,m}.$
Denote \begin{align}\label{zxcd} \textbf{a}_{1}:=(a_{1,1}, \ldots, a_{n,1})^{T}=(\mathcal{Y}\mathcal{Y}^{T})^{-1}\mathcal{Y}\textbf{z}_{1}.\end{align}
By \eqref{P1}, $\textbf{a}_{1}$ is the solution to  \eqref{dyibu} with $\textbf{z}$ being replaced by $\textbf{z}_{1}$.
Recall that $\{\textbf{y}_{1},\ldots, \textbf{y}_{n}\}$ of $\mathcal{Y}$ is a basis of $\mathbb{R}^{n}$.
Then
$\textbf{a}_{1}$ can be also expressed as $[\textbf{y}_{1},\ldots, \textbf{y}_{n}]^{-T}(z_{1,1}, \ldots, z_{n,1})^{T}$. Since the set of all the $n\times n$ invertible matrices is dense in $\mathbb{R}^{n\times n}$, there exist $\textbf{g}_{k}:=(\textbf{g}_{1,k}, \ldots, \textbf{g}_{n,k})^{T}\in \mathbb{R}^{+, n}$ for $k=2, \ldots, n$ such that
\begin{align}\notag
\mathcal{A}_{\textbf{g}}:=\left[\begin{array}{cccccccccccccccccc}
z_{1,1}&z_{1,1}+\textbf{g}_{1,2}&\cdots&z_{1,1}+\textbf{g}_{1,n}\\
z_{2,1}&z_{2,1}+\textbf{g}_{2,2}&\cdots&z_{2,1}+\textbf{g}_{2,n}\\
z_{3,1}&z_{3,1}+\textbf{g}_{3,2}&\cdots&z_{3,1}+\textbf{g}_{3,n}\\
\vdots&\vdots&\ddots&\vdots\\
z_{n,1}&z_{n,1}+\textbf{g}_{n,2}&\cdots&z_{n,1}+\textbf{g}_{n,n}
\end{array}\right]
\end{align}
is invertible
and
\begin{align} \label{fft} \begin{array}{lllllllllllllllll}\max\{||[\textbf{y}_{1},\ldots, \textbf{y}_{n}]^{-T}\textbf{g}_{k}||_{\infty}: k=2,\ldots, n\}\\
<\frac{\min\{z_{n+1,1}, \ldots, z_{m,1}\}}{||[\textbf{y}_{n+1},\ldots, \textbf{y}_{m}]^{T}||_{\infty}}.\end{array}\end{align}
For $k=2, \ldots, n$,
define \begin{align}\label{maoc} \textbf{a}_{k}:=[\textbf{y}_{1},\ldots, \textbf{y}_{n}]^{-T}((z_{1,1}, \ldots, z_{n,1})^{T}+\textbf{g}_{k}).\end{align}
Now it follows from \eqref{zxcd}, \eqref{fft} and \eqref{maoc} that \begin{align} \label{haixiang} [\textbf{z}_{1}, \ldots, \textbf{z}_{n}]:=\mathcal{Y}^{T}[\textbf{a}_{1}, \ldots, \textbf{a}_{n}]\in \mathbb{R}^{+,m}\times \mathbb{R}^{+,n}.\end{align}
That is, $\textbf{z}_{k}\in \mathcal{R}(\mathcal{Y}^{T})\cap \mathbb{R}^{+,m}$.
Using \eqref{maoc} again, the invertible matrix $\mathcal{A}_{\textbf{g}}$ consist of  the first $n$ rows of $[\textbf{z}_{1}, \ldots, \textbf{z}_{n}]$.
 Thus $\hbox{rank}([\textbf{z}_{1}, \ldots, \textbf{z}_{n}])=n$, and the proof is concluded.
\end{proof}

We also need circulant matrices that ensure fast computation (More details about this topic can be referred to \cite{Gray}).
For a   vector $\textbf{p}=(p_{0}, \ldots, p_{n-1})^{T}\in \mathbb{C}^{n}$, its  discrete Fourier transform (DFT) $\widehat{\textbf{p}}=(\widehat{p}_{0}, \ldots, \widehat{p}_{n-1})^{T}$ is defined by
$\widehat{p}_{k}=\sum^{n-1}_{l=0}p_{l}e^{\frac{-\textsf{i}2lk\pi}{n}}$.  For the row vector $\textbf{p}^{T}$, we denote its generating circulant matrix by $\hbox{circ}(\textbf{p}^{T})$, namely,
$$\hbox{circ}(\textbf{p}^{T})=\left[\begin{array}{lllllllllllllllll}
p_{0}&p_{1}&p_{2}&\cdots&p_{n-1}\\
p_{n-1}&p_{0}&p_{1}&\cdots&p_{n-2}\\
p_{n-2}&p_{n-1}&p_{0}&\cdots&p_{n-3}\\
\vdots&\vdots&\vdots&\ddots&\vdots\\
p_{1}&p_{2}&\cdots&\cdots&p_{0}
\end{array}\right].$$
The circulant matrix $\hbox{circ}(\textbf{p}^{T})$ can be decomposed by DFT via \begin{align}\label{decomposition} \hbox{circ}(\textbf{p}^{T})=nF\hbox{diag}(\widehat{p}_{0}, \cdots, \widehat{p}_{n-1})F^{*},\end{align} where    $F$ is the scaled   DFT matrix
\begin{align}\notag\begin{array}{lllllllllllllllll}
F=\\
\frac{1}{n}\left[\begin{array}{lllllllllllllllll}
1&1&1&\cdots&1\\
1&W&W^{2}&\cdots&W^{n-1}\\
1&W^{2}&W^{2\times 2}&\cdots&W^{2\times (n-1)}\\
\vdots&\vdots&\vdots&\ddots&\vdots\\
1&W^{n-1}&W^{(n-1)\times 2}&\cdots&W^{(n-1)\times (n-1)}
\end{array}\right],
\end{array}
\end{align}
with $W=e^{-\textsf{i}2\pi/n}$.
For any  $\textbf{x}\in \mathbb{R}^{n}$, by the fast Fourier transform (FFT), the computation of  $\hbox{circ}(\textbf{p}^{T})\textbf{x}$ only costs $O(n\log n)$-number of operations.
The $\ell_{0}$-norm $||\textbf{x}||_{0}$ of any vector $x$ is defined as the number of its nonzero coordinates.  
By \eqref{decomposition}, the circulant matrix  $\hbox{circ}(\textbf{p}^{T})$ is invertible if and only if $||\widehat{\textbf{p}}||_{0}=n.$

The following lemma tells us how to explicitly  construct a special anchor vector $\textbf{p}$ of $\mathcal{Y}$ in  Lemma \ref{gushi} such that $||\widehat{\textbf{p}}||_{0}=n.$ It will be seen in the proof of Theorem \ref{computational} that such an anchor  vector
is crucial for explicitly constructing a special class of  measurement vectors  that will satisfy the requirements of Theorem \ref{computational}.

\begin{lemma}\label{tjjj}  Let the frame cone $\textbf{\hbox{cone}}(\mathcal{Y})$  of $ \mathbb{R}^{n}$ be as in Lemma \ref{gushi}  such that
$\mathcal{R}(\mathcal{Y}^{T})\cap \mathbb{R}^{+,m}\neq\emptyset$. Then there exists an anchor vector  $\textsf{p}\in \hbox{invim}(\mathcal{R}(\mathcal{Y}^{T})\cap \mathbb{R}^{+,m}) $  such that $ ||\widehat{\textsf{p}}||_{0}=n$.
\end{lemma}

\begin{proof}
As in the proof of Lemma \ref{gushi},  we  assume that the first $n$-column vectors  $\{\textbf{y}_{1},\ldots, \textbf{y}_{n}\}$ of $\mathcal{Y}$ form a basis of $\mathbb{R}^{n}$.
For convenient narration, denote $\mathcal{Y}_{n}:=[\textbf{y}_{1},\ldots, \textbf{y}_{n}]$.
By Lemma  \ref{gushi}, there exist $n$-linearly independent vectors $\textbf{z}_{k}=(z_{1,k},  \ldots, z_{n,k},$
$ \ldots, z_{m,k})^{T}$
$\in \mathcal{R}(\mathcal{Y}^{T})\cap \mathbb{R}^{+,m}$, where  $k=1, \ldots, n.$
Define $\mathcal{K}_{n}:=[\textbf{z}_{1}, \textbf{z}_{2}, \ldots, \textbf{z}_{n}]$, and  as in \eqref{maoc},
$ \textbf{a}_{k}:=\mathcal{Y}_{n}^{-T}(z_{1,k}, \ldots, z_{n,k})^{T}.$
Then $\hbox{rank}([\textbf{a}_{1}, \ldots, \textbf{a}_{n}])=n.$ Moreover, by \eqref{P1}, $[\textbf{a}_{1}, \ldots, \textbf{a}_{n}]$
$=(\mathcal{Y}\mathcal{Y}^{T})^{-1}\mathcal{Y}\mathcal{K}_{n}$.
Therefore, $\hbox{rank}(\mathcal{Y}\mathcal{K}_{n})=n.$
Now for any fixed  $l\in \{1, 2, \ldots, n\}$, there exists a column  vector $\widehat{\textbf{a}}_{j}:=(\widehat{a}_{j,1}, \ldots, \widehat{a}_{j,n})^{T}$ of $$ [\widehat{\textbf{a}}_{1}, \ldots, \widehat{\textbf{a}}_{n}]=F[\textbf{a}_{1}, \ldots, \textbf{a}_{n}]=F(\mathcal{Y}\mathcal{Y}^{T})^{-1}\mathcal{Y}\mathcal{K}_{n}$$ such that \begin{align}\label{bukong} \widehat{a}_{j,l}\neq0.\end{align} If not, then
it is easy to conclude that $A(l,:)\mathcal{Y}\mathcal{K}_{n}=O,$ where $A(l,:)$ is the $l$-th row of    $A:=F(\mathcal{Y}\mathcal{Y}^{T})^{-1}$. From the  invertibility of $\mathcal{Y}\mathcal{K}_{n}$,
we deduce that $A(l,:)=O$, which is a contradiction with the invertibility of $A$.

Pick a vector $\widehat{\textbf{a}}_{\ell}\in \{\widehat{\textbf{a}}_{1}, \ldots, \widehat{\textbf{a}}_{n}\}$. If $||\widehat{\textbf{a}}_{\ell}||_{0}=n$, then the proof is completed by letting $\textsf{p}:= F^{*}\widehat{\textbf{a}}_{\ell}$. Otherwise, by the property \eqref{bukong},  there exists $\widehat{\textbf{a}}_{j}\in \{\widehat{\textbf{a}}_{1}, \ldots, \widehat{\textbf{a}}_{n}\}$ such that
$(\hbox{supp}(\widehat{\textbf{a}}_{\ell}))^{c}\cap\hbox{supp}(\widehat{\textbf{a}}_{j})\neq\emptyset$, where $\hbox{supp}(\widehat{\textbf{a}}_{j})$ is the support of  $\widehat{\textbf{a}}_{j}$, and $(\hbox{supp}(\widehat{\textbf{a}}_{\ell}))^{c}=\{1,2, \ldots, n\}\backslash \hbox{supp}(\widehat{\textbf{a}}_{\ell})$. It is easy to prove that $||\nu \widehat{\textbf{a}}_{\ell}+\widehat{\textbf{a}}_{j}||_{0}\geq ||\widehat{\textbf{a}}_{\ell}||_{0}+1$, where $$\nu>\max_{l\in \hbox{supp}(\widehat{\textbf{a}}_{\ell})}|\frac{\widehat{a}_{j,l}}{\widehat{a}_{\ell,l}}|.$$
On the other hand, it is obvious that  $\nu \widehat{\textbf{a}}_{\ell}+\widehat{\textbf{a}}_{j}\in F(\mathcal{Y}\mathcal{Y}^{T})^{-1}\mathcal{Y}\big(\mathcal{R}(\mathcal{Y}^{T})\cap \mathbb{R}^{+,m}\big).$
Thus by  at most $n$-procedures discussed  above, we will be able to get  a vector $\widehat{\textbf{a}}\in F(\mathcal{Y}\mathcal{Y}^{T})^{-1}\mathcal{Y}\big(\mathcal{R}(\mathcal{Y}^{T})\cap \mathbb{R}^{+,m}\big)$
such that $||\widehat{\textbf{a}}||_{0}=n.$  Therefore   \begin{align}\label{tql} \textsf{p}:=F^{*}\widehat{\textbf{a}}\end{align} is an anchor vector satisfying   $ ||\widehat{\textsf{p}}||_{0}=n$. \end{proof}

Next based on the proofs of Lemmas \ref{gushi} and \ref{tjjj}, we establish Algorithm \ref{AL2} for designing an anchor vector
$\textsf{p}\in \hbox{invim}(\mathcal{R}(\mathcal{Y}^{T})\cap \mathbb{R}^{+,m})$ such that $ ||\widehat{\textsf{p}}||_{0}=n$.

\begin{algorithm}\label{AL2}
    \caption{Based on   $\textsf{q}_{1}\in \hbox{invim}(\mathcal{R}(\mathcal{Y}^{T})\cap \mathbb{R}^{+,m})$, design  $\textsf{p}\in \hbox{invim}(\mathcal{R}(\mathcal{Y}^{T})\cap \mathbb{R}^{+,m})$ such that $ ||\widehat{\textsf{p}}||_{0}=n$. }
     \KwIn{$\mathcal{Y}=[\textbf{y}_{1},\ldots, \textbf{y}_{m}]\in \mathbb{R}^{n\times m}$, $\textsf{q}_{1}\in \hbox{invim}(\mathcal{R}(\mathcal{Y}^{T})\cap \mathbb{R}^{+,m})$,   $\textbf{z}_{1}=\mathcal{Y}^{T}
     \textsf{q}_{1}$, $\widehat{\textsf{q}}_{1}= F \textsf{q}_{1}.$}

  \If {$||\widehat{\textsf{q}}_{1}||_{0}<n$}
 {Extend $\textbf{z}_{1}$ to linearly independent vectors   $ \{\textbf{z}_{k}\}^{n}_{k=1}\subseteq \mathcal{R}(\mathcal{Y}^{T})\cap \mathbb{R}^{+,m}$
by using   \eqref{maoc} and \eqref{haixiang};
$[\widehat{\textsf{q}}_{2},  \ldots, \widehat{\textsf{q}}_{n}]\leftarrow F(\mathcal{Y}\mathcal{Y}^{T})^{-1}\mathcal{Y}[\textbf{z}_{2}, \ldots, \textbf{z}_{n}]$;

\For{$j=2:n$}
{Find $\widehat{\textsf{q}}_{l}\in \{\widehat{\textsf{q}}_{2},  \ldots, \widehat{\textsf{q}}_{n}\}$ such that $(\hbox{supp}(\widehat{\textsf{q}}_{1}))^{c}\cap\hbox{supp}(\widehat{\textsf{q}}_{l})\neq\emptyset$. Pick  $\nu>\max_{\ell\in \hbox{supp}(\widehat{\textbf{q}_{1}})}\frac{|\widehat{\textbf{q}_{l}}_{,\ell}|}{|\widehat{\textbf{q}_{1}}_{,\ell}|}$;
$\widehat{\textbf{q}_{1}}\gets\nu \widehat{\textbf{q}_{1}}+\widehat{\textbf{q}_{l}}$;

\If{$||\widehat{\textbf{q}_{1}}||_{0}=n$}
{break;}
  }}

\KwOut{$\textsf{p}=F^{*}\widehat{\textbf{q}}_{1}$. }
\end{algorithm}

\subsection{Proof of Theorem \ref{computational}}\label{opq}
The proof will be  concluded from two cases: frame cone and non-frame cone.
\subsubsection{\textbf{$\textbf{\hbox{cone}}(\mathcal{X})$ is a frame cone}}\label{dshuj}

Obviously,  $\gamma=\hbox{rank}(\mathcal{X})=n.$
By Algorithm \ref{AL2}, we can construct  an anchor vector  $\textsf{p}_{1}\in \hbox{invim}(\mathcal{R}(\mathcal{X}^{T})\cap \mathbb{R}^{+,m}) $ such that \begin{align} \label{anchor}||\widehat{\textsf{p}_{1}}||_{0}=n.\end{align}
 Thus the circulant matrix $\hbox{circ}(\textsf{p}^{T}_{1})$ is invertible. Denote $\hbox{circ}(\textsf{p}^{T}_{1})=[\textsf{p}_{1}, \textsf{p}_{2}, \ldots, \textsf{p}_{n}]^{T}$.  Let $\textsf{f}_{1}:=\textsf{p}_{1}$
 and design  $\{\textsf{f}_{k}\}^{n}_{k=2}$ by
\begin{align}\label{celxl}
\textsf{f}_{k}=\delta_{k}\textsf{p}_{1}+\textsf{p}_{k}, k\geq2,
\end{align}
where $\delta_{k}>0$ is selected in such a way that any $\textbf{x}_{l}\in \{\textbf{x}_{1},\ldots, \textbf{x}_{m}\}$
satisfies
\begin{align}
\label{df1} \langle \textbf{x}_{l},  \textsf{f}_{k}\rangle>0.
\end{align}

It follows from   \eqref{df1} that
$\hbox{sgn}(\langle \textbf{z}, \textsf{f}_{k}\rangle)\geq0$ for any $\textbf{z}\in \textbf{\hbox{cone}}(\mathcal{X})$ and $k=2, \ldots, n.$ On the other hand,  it is easy to follow from
\begin{align}\label{fcheng00} \left[\begin{array}{lllllllllllllllll}
1&0&0&\cdots&0\\
\delta_{2}&1&0&\cdots&0\\
\delta_{3}&0&1&\cdots&0\\
\vdots&\vdots&\vdots&\ddots&\vdots\\
\delta_{n}&0&0&\cdots&1
\end{array}\right]\left[\begin{array}{lllllllllllllllll}
\textsf{p}^{T}_{1}\\
\textsf{p}^{T}_{2}\\
\textsf{p}^{T}_{3}\\
\vdots\\
\textsf{p}^{T}_{n}
\end{array}\right]=\left[\begin{array}{lllllllllllllllll}
\textsf{f}^{T}_{1}\\
\textsf{f}^{T}_{2}\\
\textsf{f}^{T}_{3}\\
\vdots\\
\textsf{f}^{T}_{n}
\end{array}\right]\end{align}
that $\{\textsf{f}_{k}\}^{n}_{k=1}$ is a basis  of $\mathbb{C}^{n}$. Thus the  target  signal $\textbf{z}\in \textbf{\hbox{cone}}(\mathcal{X})$ can be determined, up to a global sign, by the following linear system of  equations
\begin{align}\notag \left[\begin{array}{lllllllllllllllll}
1&0&0&\cdots&0\\
\delta_{2}&1&0&\cdots&0\\
\delta_{3}&0&1&\cdots&0\\
\vdots&\vdots&\vdots&\ddots&\vdots\\
\delta_{n}&0&0&\cdots&1
\end{array}\right]\left[\begin{array}{lllllllllllllllll}
\textsf{p}^{T}_{1}\\
\textsf{p}^{T}_{2}\\
\textsf{p}^{T}_{3}\\
\vdots\\
\textsf{p}^{T}_{n}
\end{array}\right]\textbf{z}=\left[\begin{array}{lllllllllllllllll}
|\langle \textbf{z},  \textsf{f}_{1}\rangle|\\
|\langle \textbf{z},  \textsf{f}_{2}\rangle|\\
|\langle \textbf{z},  \textsf{f}_{3}\rangle|\\
\quad  \vdots\\
|\langle \textbf{z},  \textsf{f}_{n}\rangle|
\end{array}\right].\end{align}
By \eqref{decomposition}, the above system  can be rewritten as
\begin{align}\notag \begin{array}{ll} n\left[\begin{array}{lllllllllllllllll}
1&0&0&\cdots&0\\
\delta_{2}&1&0&\cdots&0\\
\delta_{3}&0&1&\cdots&0\\
\vdots&\vdots&\vdots&\ddots&\vdots\\
\delta_{n}&0&0&\cdots&1
\end{array}\right]F\hbox{diag}(\widehat{p}_{0}, \cdots, \widehat{p}_{n-1})F^{*}\textbf{z}\\
=\left[\begin{array}{lllllllllllllllll}
|\langle \textbf{z},  \textsf{f}_{1}\rangle|\\
|\langle \textbf{z},  \textsf{f}_{2}\rangle|\\
|\langle \textbf{z},  \textsf{f}_{3}\rangle|\\
\quad  \vdots\\
|\langle \textbf{z},  \textsf{f}_{n}\rangle|
\end{array}\right].\end{array}\end{align}
That is, up to a global sign, $\textbf{z}$ can be recovered by
\begin{align}\label{fastchonggou}\begin{array}{lllllllllllllllll}
\textbf{z}=\\
\textcolor[rgb]{0.00,0.07,1.00}{\hbox{FFT}}\Big(\hbox{diag}^{-1}(\textcolor[rgb]{0.00,0.07,1.00}{\hbox{FFT}}(\textsf{p}^{T}_{1}))\textcolor[rgb]{0.00,0.07,1.00}{\hbox{IFFT}}\Big(\left[\begin{array}{lllllllllllllllll}
1&0&\cdots&0\\
-\delta_{2}&1&\cdots&0\\
\vdots&\vdots&\ddots&\vdots\\
-\delta_{n}&0&\cdots&1
\end{array}\right]\\
\times \left[\begin{array}{lllllllllllllllll}
|\langle \textbf{z},  \textsf{f}_{1}\rangle|\\
|\langle \textbf{z},  \textsf{f}_{2}\rangle|\\
\quad \ \vdots\\
|\langle \textbf{z},  \textsf{f}_{n}\rangle|
\end{array}\right]\Big)\Big).
\end{array}
\end{align}
It is easy to see that the computational complexity of \eqref{fastchonggou} is $O(n\log n)$.

\subsubsection{\textbf{$\textbf{\hbox{cone}}(\mathcal{X})$ is not  a frame cone}}\label{dashujiao}
Denote $\gamma:=\hbox{rank}(\mathcal{X}).$
Then $\gamma<n.$
Define an isometry $\mathfrak{P}: \hbox{span}\{\textbf{x}_{1},\ldots, \textbf{x}_{m}\}\longrightarrow \mathbb{R}^{\gamma}$.
Specifically, \begin{align}\label{isometry} \begin{array}{lllllllllllllllll} \mathfrak{P}(\widetilde{\textbf{e}}_{k})=\textbf{e}_{k}, k=1, \ldots, \gamma,\end{array}\end{align}
where $\{\widetilde{\textbf{e}}_{k}\}^{\gamma}_{k=1}$ and $\{\textbf{e}_{k}\}^{\gamma}_{k=1}$ are the orthornormal basis and the standard orthornormal  basis of $\hbox{span}\{\textbf{x}_{1},\ldots, \textbf{x}_{m}\}$ and $\mathbb{R}^{\gamma}$, respectively. Denote $\mathcal{Y}:=\mathfrak{P}\mathcal{X}$.
By the linear and isometry   property,  $\mathfrak{P}(\textbf{\hbox{cone}}(\mathcal{X}))=\textbf{\hbox{cone}}(\mathcal{Y})$, and $\mathcal{Y}$
also satisfies the overlap property \eqref{ufff}. By Algorithm \ref{AL2}, we can design an anchor vector $\widetilde{\textsf{p}}_{1}\in \mathbb{R}^{\gamma}$ of $\mathcal{Y}$
such that $\mathcal{Y}^{T}\widetilde{\textsf{p}}_{1}\succ0$ and $||\widehat{\widetilde{\textsf{p}}_{1}}||_{0}=\gamma.$

Denote $\widetilde{\textsf{f}}_{1}:=\widetilde{\textsf{p}}_{1}\in \mathbb{R}^{\gamma}.$
Invoking Case \ref{dshuj} for $n=\gamma$, we can additionally  design $(\gamma-1)$  vectors $\{\widetilde{\textsf{f}}_{k}\}^{\gamma}_{k=2}$ such that $\{\widetilde{\textsf{f}}_{k}\}^{\gamma}_{k=1}$ are phase retrievable
for $\textbf{\hbox{cone}}(\mathcal{Y})$.
That is, any signal  $\widetilde{\textbf{z}}\in \textbf{\hbox{cone}}(\mathcal{Y})$ can be determined by the $\gamma$ magnitude
measurements $\{|\langle \widetilde{\textbf{z}}, \widetilde{\textsf{f}}_{k} \rangle|\}^{\gamma}_{k=1}$, and
the corresponding complexity is $O(\gamma\log \gamma)$. Particularly, for the target $\textbf{z}$, its projection  $\mathfrak{P}\textbf{z}$ can be recovered by invoking \eqref{fastchonggou}, namely,
\begin{align}\label{fastchonggoulow12}\begin{array}{lllllllllllllllll}
\mathfrak{P}\textbf{z}=\textcolor[rgb]{0.00,0.07,1.00}{\hbox{FFT}}\Big(\hbox{diag}^{-1}(\textcolor[rgb]{0.00,0.07,1.00}{\hbox{FFT}}(\widetilde{\textsf{p}}^{T}_{1}))\\
\times \textcolor[rgb]{0.00,0.07,1.00}{\hbox{IFFT}}\Big(\left[\begin{array}{lllllllllllllllll}
1&0&\cdots&0\\
-\delta_{2}&1&\cdots&0\\
\vdots&\vdots&\ddots&\vdots\\
-\delta_{\gamma}&0&\cdots&1
\end{array}\right]
\left[\begin{array}{lllllllllllllllll}
|\langle \mathfrak{P}\textbf{z},  \widetilde{\textsf{f}}_{1}\rangle|\\
|\langle \mathfrak{P}\textbf{z},  \widetilde{\textsf{f}}_{2}\rangle|\\
\quad \ \vdots\\
|\langle \mathfrak{P}\textbf{z},  \widetilde{\textsf{f}}_{\gamma}\rangle|
\end{array}\right]\Big)\Big),
\end{array}
\end{align}
where the constants $\{\delta_{k}\}^{\gamma}_{k=2}$ satisfy  \eqref{df1} with $n$, $\mathcal{X}$ and $\textsf{p}_{1}$ being replaced by
$\gamma$, $\mathcal{Y}$ and $\widetilde{\textsf{p}}_{1}$, respectively. Now define
$\textbf{f}_{k}:=\mathfrak{P}^{-1}\widetilde{\textbf{f}}_{k}, k=1, \ldots, \gamma.$ By the isometry
property, we have  $|\langle  \mathfrak{P}\textbf{z}, \widetilde{\textsf{f}}_{k}\rangle|=|\langle \textbf{z}, \textsf{f}_{k}\rangle|$.
Then the recovery formula \eqref{fastchonggoulow12} can be rewritten as
\begin{align}\label{fastchogngouxy1}\begin{array}{lllllllllllllllll}
\mathfrak{P}\textbf{z}\\
=\Big[\textcolor[rgb]{0.00,0.07,1.00}{\hbox{FFT}}\Big(\hbox{diag}^{-1}(\textcolor[rgb]{0.00,0.07,1.00}{\hbox{FFT}}((\mathfrak{P}\textsf{f}_{1})^{T}))\\
\\\times \textcolor[rgb]{0.00,0.07,1.00}{\hbox{IFFT}}\Big(\left[\begin{array}{lllllllllllllllll}
1&0&0&\cdots&0\\
-\delta_{2}&1&0&\cdots&0\\
\vdots&\vdots&\vdots&\ddots&\vdots\\
-\delta_{\gamma}&0&0&\cdots&1
\end{array}\right]\left[\begin{array}{lllllllllllllllll}
|\langle \textbf{z},  \textsf{f}_{1}\rangle|\\
|\langle \textbf{z},  \textsf{f}_{2}\rangle|\\
\quad \ \vdots\\
|\langle \textbf{z},  \textsf{f}_{\gamma}\rangle|
\end{array}\right]\Big)\Big)\Big],\end{array}
\end{align}
Denote $\mathfrak{P}\textbf{z}=\sum^{\gamma}_{k=1}c_{k}\textbf{e}_{k}$.
Then \begin{align}\notag  \begin{array}{lllllllllllllllll}\textbf{z}=\mathfrak{P}^{-1}\mathfrak{P}\textbf{z}=\sum^{\gamma}_{k=1}c_{k}\widetilde{\textbf{e}}_{k}\end{array}\end{align} which costs $ O(\gamma)$ operations.
 Then the total complexity is
$ O(\gamma\log \gamma)+O(\gamma)=O(\gamma\log \gamma).$
Integrating  Subsection \ref{dshuj} and \ref{dashujiao}, the proof is concluded. \hfill $\blacksquare$

\subsection{Algorithm for designing measurement vectors for a cone  satisfying the overlap property \eqref{ufff}}
Based on Algorithm \ref{AL2} and Subsection \ref{opq} (the proof of   Theorem \ref{computational}), we propose the following Algorithm \ref{AL3} for explicitly constructing    $\hbox{rank}(\mathcal{X})$-vectors that can be used to perform the  fast phase-retrieval for $\textbf{\hbox{cone}}(\mathcal{X})$.


\begin{algorithm}\label{AL3}
    \caption{Designing $\hbox{rank}(\mathcal{X})$-vectors for the  fast phase-retrieval of  $\textbf{\hbox{cone}}(\mathcal{X})$ satisfying the overlap property \eqref{ufff}. }
     \KwIn{$\mathcal{X}=[\textbf{x}_{1},\ldots, \textbf{x}_{m}]\in \mathbb{R}^{n\times m}$, $\textsf{q}_{1}\in \hbox{invim}(\mathcal{R}(\mathcal{X}^{T})\cap \mathbb{R}^{+,m})$, $\gamma=\hbox{rank}(\mathcal{X})$,  and an  isometry (an $\gamma\times n$ matrix)
     $\mathfrak{P}: \hbox{span}\{\textbf{x}_{1},\ldots, \textbf{x}_{m}\}\longrightarrow \mathbb{R}^{\gamma}$.
      \% If $\gamma=n$, then we just pick $\mathfrak{P}$ as the identity matrix.\% }

     $\mathcal{Y}\gets\mathfrak{P}\mathcal{X}$; $\textsf{q}_{1}\gets\mathfrak{P}\textsf{q}_{1}$.

If $||\widehat{\textsf{q}_{1}}||_{0}<\gamma$, then
using $\textsf{q}_{1}$,
design an anchor vector $\widetilde{\textsf{p}}_{1}\in \hbox{invim}(\mathcal{R}(\mathcal{Y}^{T})\cap \mathbb{R}^{+,m})$ by Algorithm \ref{AL2} such that
$||\widehat{\widetilde{\textsf{p}}_{1}}||_{0}=\gamma.$

 Construct a circulant matrix $[\widetilde{\textsf{p}}_{1}, \widetilde{\textsf{p}}_{2}, \ldots, \widetilde{\textsf{p}}_{n}]^{T}=\hbox{circ}(\widetilde{\textsf{p}}^{T}_{1})$.
Let  $\widetilde{\textbf{f}}_{1}:= \widetilde{\textsf{p}}_{1}$ and design
\begin{align}\label{hq}\widetilde{\textbf{f}}_{k}:=\delta_{k}\widetilde{\textsf{p}}_{1}+\widetilde{\textsf{p}}_{k}, k\geq2,\end{align}
where  $\{\delta_{k}\}^{\gamma}_{k=2}$ is chosen appropriately  such that  \eqref{df1} holds with $n$, $\mathcal{X}$ and $\textsf{p}_{1}$ being replaced by
$\gamma$, $\mathcal{Y}$ and $\widetilde{\textsf{p}}_{1}$, respectively.

\KwOut{\begin{align}\label{mjia}\textbf{f}_{k}\gets \mathfrak{P}^{-1}\widetilde{\textbf{f}}_{k}, k=1, \ldots, \gamma. \end{align}}
\end{algorithm}

The existence of $\{\delta_{k}\}^{\gamma}_{k=2}$ in Algorithm \ref{AL2}  is guaranteed by the following remark.
\begin{remark}
There are many choices for the sequence $\{\delta_{k}\}^{\gamma}_{k=2}$ in \eqref{hq}.
For example, for any $k\in \{2, \ldots, m\}$,
if   \begin{align}\label{ff} \delta_{k}>\frac{||\widetilde{\textsf{p}}_{k}||_{2}\max\{||\textbf{y}_{i}||_{2}: i=1, \ldots, m\}}{\kappa_{\min}},\end{align}
 where $\kappa_{\min}=\min\{\langle \textbf{x}_{1}, \textsf{p}_{1}\rangle, \ldots, \langle \textbf{x}_{m}, \textsf{p}_{1}\rangle\}$,
then  for any $\textbf{x}_{l}\in \{\textbf{x}_{1},\ldots, \textbf{x}_{m}\}$  it follows from \eqref{hq} and \eqref{ff} that
 \begin{align}\begin{array}{lllll}\notag \langle \textbf{x}_{l},  \textsf{f}_{k}\rangle&=\langle \textbf{y}_{l},  \widetilde{\textsf{f}}_{k}\rangle\\
 &\geq \delta_{k}\langle\textbf{y}_{l},\widetilde{\textsf{p}}_{1}\rangle-|\langle
 \textbf{y}_{l}, \widetilde{\textsf{p}}_{k}\rangle|\\
 &=\delta_{k}\langle\textbf{x}_{l},\textsf{p}_{1}\rangle-|\langle
 \textbf{y}_{l}, \widetilde{\textsf{p}}_{k}\rangle|\\
 &\geq\delta_{k}\kappa_{\min}-||\widetilde{\textsf{p}}_{k}||_{2}\max\{||\textbf{y}_{i}||_{2}: i=1, \ldots, m\}\\
 &\geq0.\hfill \blacksquare
 \end{array}\end{align}
\end{remark}
\subsection{Recovery formula}\label{chonggougongshi}
In this subsection we abstract the recovery formula from Subsection \ref{opq}.
Suppose that the target  $\textbf{z}$ lies  in the detectable UoC $\bigcup^{L}_{l=1} \textbf{\hbox{cone}}(\mathcal{X}_{l})$.
After the detection we find that  $\textbf{z}\in \textbf{\hbox{cone}}(\mathcal{X}_{k})$. If $\textbf{\hbox{cone}}(\mathcal{X}_{k})$ satisfies \eqref{ufff},
then $\textbf{z}$ can be recovered by the following two procedures:

\textbf{P1}:
\begin{align}\label{fastchogngouxy}\begin{array}{lll}\begin{array}{lllllllllllllllll}
\mathfrak{P}\textbf{z}\\
:=\sum^{\gamma}_{k=1}c_{k}\textbf{e}_{k}\\
=\Big[\textcolor[rgb]{0.00,0.07,1.00}{\hbox{FFT}}\Big(\hbox{diag}^{-1}(\textcolor[rgb]{0.00,0.07,1.00}{\hbox{FFT}}((\mathfrak{P}\textsf{f}_{1})^{T}))\\
\times \textcolor[rgb]{0.00,0.07,1.00}{\hbox{IFFT}}\Big(\left[\begin{array}{lllllllllllllllll}
1&0&0&\cdots&0\\
-\delta_{2}&1&0&\cdots&0\\
\vdots&\vdots&\vdots&\ddots&\vdots\\
-\delta_{\gamma}&0&0&\cdots&1
\end{array}\right]\left[\begin{array}{lllllllllllllllll}
|\langle \textbf{z},  \textsf{f}_{1}\rangle|\\
|\langle \textbf{z},  \textsf{f}_{2}\rangle|\\
\quad \ \vdots\\
|\langle \textbf{z},  \textsf{f}_{\gamma}\rangle|
\end{array}\right]\Big)\Big)\Big].\end{array}
\end{array}
\end{align}

\textbf{P2}:
\begin{align}\label{final} \textbf{z}=\mathfrak{P}^{-1}\mathfrak{P}f=\sum^{\gamma}_{k=1}c_{k}\widetilde{\textbf{e}}_{k}.\end{align}

The following note    is helpful  for  conducting    \eqref{fastchogngouxy} and \eqref{final}.
\begin{note}\label{note}
(i) $\gamma=\hbox{rank}(\mathcal{X}_{k})$. (ii) $\{\widetilde{\textbf{e}}_{k}\}^{\gamma}_{k=1}$
and $\{\textbf{e}_{k}\}^{\gamma}_{k=1}$ are the orthornormal basis and the standard orthornormal  basis of $\hbox{span}\{\textbf{x}_{k,1},\ldots, \textbf{x}_{k,m_{k}}\}$ and $\mathbb{R}^{\gamma}$, respectively. The map $\mathfrak{P}: \hbox{span}\{\textbf{x}_{k,1},\ldots, \textbf{x}_{k,m}\}\longrightarrow \mathbb{R}^{\gamma}$
is an  isometry (an $\gamma\times n$ matrix). As in  Algorithm \ref{AL3} we just set $\mathfrak{P}$ to the identity matrix
when $\gamma=n$. (iii)
The  measurement vectors $\{\textsf{f}_{k}\}^{\gamma}_{k=2}$
are designed by Algorithm \ref{AL3} with $\mathcal{X}$ therein  being replaced by $\mathcal{X}_{k}$,
and the   sequence $\{\delta_{k}\}^{\gamma}_{k=2}$ satisfies the requirement in \eqref{hq}. \hfill $\blacksquare$
\end{note}

\subsection{Stability of the recovery formula \eqref{fastchogngouxy} and \eqref{final}}
Since the measurements are often  contaminated by noise in practice, we need to   establish the stability for the recovery in Subsection \ref{chonggougongshi} (\eqref{fastchogngouxy} and \eqref{final})  in the noisy setting.
We will  consider the model for  observing a measurement in the noisy setting:
\begin{align} \label{noisemodel}  \begin{array}{lllllllllllllllll}\widetilde{|\langle \textbf{q}, \textbf{z}\rangle|}=|\langle \textbf{q}, \textbf{z}\rangle|+\underline{\hbox{n}},\end{array}\end{align}
where $\textbf{q}$ represents any measurement vector and the additive  noise $\underline{\hbox{n}}$
obeys the Gaussian distribution, namely,
\begin{align}\label{noisezhengtai}\begin{array}{lllllllllllllllll} \underline{\hbox{n}}\sim \emph{\textbf{N}}(0,\sigma^{2}).\end{array}\end{align}
The chi-square distribution $\chi^{2}(s)$  with $s$ degrees  of  freedom will be  useful for probability estimation.
Its   density function is
$$\rho_{s}(x)=\left\{\begin{array}{lllllllllllllllll}
\frac{1}{2^{s/2}G(\frac{s}{2})}x^{s/2-1}e^{-x/2},& x>0,\\
0,& x\leq0,
\end{array}\right.$$
with the Gamma function $G(t):=\int^{\infty}_{0}x^{t-1}e^{-x}dx.$ Denote the distribution function by $\Phi_{s}(t):=\int^{t}_{-\infty}\rho_{s}(x) dx$.

 \begin{figure}\label{Figurecvx}
    \centering
\includegraphics[width=9cm, height=6cm]{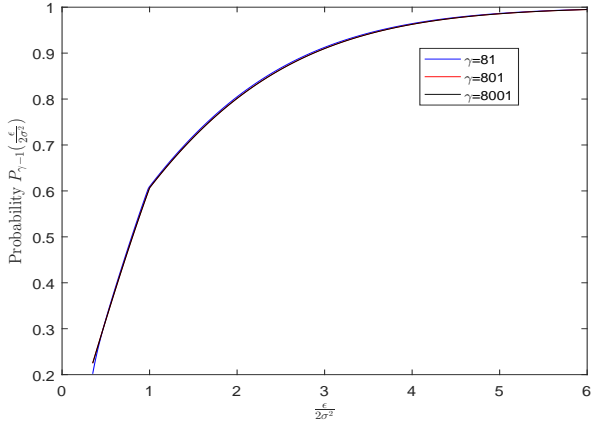}
    \caption{The probability  $P_{\gamma-1}(\frac{\epsilon}{2\sigma^{2}})$ in \eqref{prob} vs $\frac{\epsilon}{2\sigma^{2}}$.}
\end{figure}

\begin{theorem}\label{ghhh}
Let the target   $\textbf{z}\in \textbf{\hbox{cone}}(\mathcal{X}_{k})$ be as in Subsection \ref{chonggougongshi}.
Consequently, it can be recovered by \eqref{fastchogngouxy} and \eqref{final}. Suppose that the measurements
 $|\langle \textbf{f}_{k}, \textbf{z}\rangle|$ used for \eqref{fastchogngouxy} is  contaminated by the  noise $\underline{\hbox{n}}_{k}$ obeying the Gaussian  distribution in \eqref{noisezhengtai},
 $k=1, \ldots, \gamma$.
Then for any fixed $\epsilon>0$, with at least the following probability
\begin{align}\label{prob}\begin{array}{lllllllllllllllll}
P_{\gamma-1}(\frac{\epsilon}{2\sigma^{2}})\\
=-1+\Phi_{\gamma-1}(\gamma-1+\frac{\gamma \epsilon}{2\sigma^{2}})+\Phi_{1}(1+\frac{\gamma\epsilon}{2\sigma^2(\gamma-1)})\\
\quad -\Phi_{\gamma-1}(\gamma-1-\frac{\gamma \epsilon}{2\sigma^{2}})-\Phi_{1}(1-\frac{\gamma\epsilon}{2\sigma^2(\gamma-1)}),
\end{array}\end{align}
the recovery error is bounded by
\begin{align}\label{noiseerror}\begin{array}{lllllllllllllllll}
\min\{||\textbf{z}-\widetilde{\textbf{z}}_r||_{2}, ||\textbf{z}+\widetilde{\textbf{z}}_r||_{2}\}\\
\leq \frac{\sqrt{2||\underline{\textbf{n}}||^{2}_2+\max\{\delta_{2}, \ldots, \delta_{\gamma}\}[(\gamma-1)\epsilon+\frac{\gamma-1}{\gamma}||\underline{\textbf{n}}||^{2}_2]}}{\min|\textcolor[rgb]{0.00,0.07,1.00}{\hbox{FFT}}(\mathfrak{P}\textsf{f}_{1})|},
\end{array}
\end{align}
where $\underline{\textbf{n}}=(\underline{\hbox{n}}_{1}, \ldots, \underline{\hbox{n}}_{\gamma})$ and $\widetilde{\textbf{z}}_r$ is the recovery result from \eqref{fastchogngouxy}
 and \eqref{final} in the noisy setting.
\end{theorem}
\begin{proof}
The proof is given in the Appendix section.
\end{proof}

The graphs of $P_{\gamma-1}(\frac{\epsilon}{2\sigma^{2}})$ in \eqref{prob} corresponding to $\gamma=81, 801$ and $ 8001$
are plotted in Fig. III.1. It is observed in   Fig. III.1 that   as  $\gamma=\hbox{rank}(\mathcal{X})$ increases,
the behavior of $P_{\gamma-1}(\frac{\epsilon}{2\sigma^{2}})$ changes very mildly.

\section{Numerical simulation}\label{numerical}
We have  established in Section \ref{section2}  the two-step
 PR-scheme  for detectable UoCs (Algorithm \ref{separating} for detection while   \eqref{fastchogngouxy} and \eqref{final} for recovery). As mentioned in Proposition  \ref{complex}, it   requires very few measurements and has low  computational complexity.
On the other hand, as introduced in Section \ref{section1}  some efficient phase retrieval methods   are available  in the literature.
As an iterative method, Alternating Minimization \cite{Netrapalli} converges geometrically to the target, and shows good performance on recovery accuracy.
BlockPR \cite{WY} performs well on the aspect of computational speed.
The task   of this  section  is to
present some numerical simulations demonstrating   the efficiency of the two-step PR-scheme, and to  compare  with
Alternating Minimization  and  BlockPR   on  the aspects  of  time cost, measurement  cost (the amount of
measurements) and   relative error (recovery accuracy).

\subsection{Two-step PR-scheme  for random signals  in the noiseless setting}\label{noiseless}
Let $\mathcal{X}_{1}:=[\textbf{x}_{1,1},\ldots,$
$ \textbf{x}_{1,2n-1}]\in\mathbb{R}^{n\times (2n-1)}.$ Herein
\begin{align} \begin{array}{lllllllllllllllll} \textbf{x}^{T}_{1,1}=\Big[1,\frac{-1}{3\times 2^{3}\times 1}, \frac{1}{3\times 3^{3}\times 2}, \ldots, \frac{(-1)^{l-1}}{3\times l^{3}\times (l-1)}, \ldots, \\
\frac{(-1)^{n-1}}{3\times n^{3}\times (n-1)}\Big],\end{array}\end{align}
and  for $2\leq k\leq n$, $\textbf{x}^{T}_{1,k}=\textbf{x}^{T}_{1,1}\circ \hbox{\textbf{1}}_{k}$, where  $\hbox{\textbf{1}}_{k}:=[1, \ldots, 1, -1, 1, \ldots, 1]$ with $-1$ being the $k$-th element,
$\circ$ is  the element-wise product  of two vectors, and
\begin{align}\notag \begin{array}{lllllllllllllllll} [\textbf{x}_{1,n+1},\ldots, \textbf{x}_{1,2n-1}]\\
:=[\textbf{x}_{1,1},\ldots, \textbf{x}_{1,n}]\left[\begin{array}{cccccccccccccccccc}
b&b&\cdots&b\\
-a&0&\cdots&0\\
0&-a&\cdots&0\\
\vdots&\vdots&\ddots&\vdots\\
0&0&\cdots&-a
\end{array}\right]_{n\times (n-1),} \end{array}\end{align}
with $a=0.115$, $b=0.8850$.
Furthermore define
\begin{align}\notag \begin{array}{lllllllllllllllll} \mathcal{X}_{2}:=[\textbf{x}_{2,1},\ldots, \textbf{x}_{2,n}]\\
=\left[\begin{array}{cccccccccccccccccc}
2&2&\cdots&2\\
-1&-1&\cdots&-1\\
(-1)^{3+1}&(-1)^{3+2}&\cdots&(-1)^{3+n}\\
\vdots&\vdots&\ddots&\vdots\\
(-1)^{n+1}&(-1)^{n+2}&\cdots&(-1)^{n+n}
\end{array}\right]_{n\times n.}
\end{array}\end{align}
Pick   $\textbf{g}:=[1, 2, 0, \ldots, 0]^{T}.$ 
By the  direct computation, we found  that  \begin{align}\label{dcg}\begin{array}{lllllllllllllllll}  \mathcal{X}^{T}_{1}\textbf{g}\succ\textbf{0}, \mathcal{X}^{T}_{2}\textbf{g}=\textbf{0}.  
\end{array}\end{align}
That is, $\bigcup^{2}_{k=1} \textbf{\hbox{cone}}(\mathcal{X}_{k})$ is detectable and $\textbf{g}$ is eligible  as a detector. Moreover it is easy to check that
both $\textbf{\hbox{cone}}(\mathcal{X}_{1})$ and $\textbf{\hbox{cone}}(\mathcal{X}_{2})$
satisfy the  property  \eqref{ufff}.

Any   signal  $\textbf{z}\in \bigcup^{2}_{k=1} \textbf{\hbox{cone}}(\mathcal{X}_{k})$ can be reconstructed  by
 the scheme \eqref{step} which is conducted via    Algorithm \ref{separating} and the  formulas \eqref{fastchogngouxy}, \eqref{final}.
Pick the target random  signal  \begin{align}\label{target} \begin{array}{lllllllllllllllll}   \textbf{z}=\sum^{2n-1}_{k=1}\underline{\epsilon}_{k}\textbf{x}_{1,k}\in \textbf{\hbox{cone}}(\mathcal{X}_{1})\end{array}\end{align}
 as an example to check the   efficiency of  \eqref{step}, where the random variable $\underline{\epsilon}_{k}$
obeys  the uniform distribution on the interval $(0, 1/100)$. In this case, $\gamma=\hbox{rank}(\mathcal{X}_{1})=n $ and as mentioned in
 Note \ref{note}, the  isometry $\mathfrak{P}$ in \eqref{fastchogngouxy} and \eqref{final} is set to the identity matrix.
 There are many choices of $\textbf{q}_{1}\in \hbox{invim}(\mathcal{R}(\mathcal{X}_{1}^{T})\cap \mathbb{R}^{+,2n-1})$ in Algorithm \ref{AL3}.
 For example
 choose $\textbf{q}_{1}=(1, 0, 0, \ldots, 0)^{T}$.
Based on $\textbf{q}_{1}$,  we design $\{\textbf{f}_{k}\}^{n}_{k=1}$ by  using  Algorithm \ref{AL3} such that
 $\textbf{f}_{1}=\textbf{q}_{1}$ and $\{\textbf{f}_{k}\}^{n}_{k=2}$ are  given by \eqref{hq} with  $\delta_k=0.0542$.
 By direct computation, we can check that both \eqref{anchor} and \eqref{df1} hold  with $\textbf{p}_{1}$ therein  being replaced by $\textbf{f}_{1}$. Therefore the $n+1$  vectors $\{\textbf{g}, \textbf{f}_{1}, \ldots, \textbf{f}_{n}\}$ are phase retrievable for the target signal.
Specifically, Algorithm \ref{separating}  is conducted by using the detector  $\textbf{g}$.
After the detection we found $\textbf{z}\in \textbf{\hbox{cone}}(\mathcal{X}_{1})$,
and the recovery  formula  \eqref{fastchogngouxy} is  conducted   by the magnitude measurements
$\{|\langle \textbf{z}, \textbf{f}_{1}\rangle|, \ldots, |\langle \textbf{z}, \textbf{f}_{n}\rangle|\}$.

 We include the simulation results  of
Alternating Minimization  and  BlockPR  for comparison.
 Incidentally, we use the matlab  software available in \cite{softwareforpr}
to conduct   the BlockPR.
 The relative recovery  error is  defined by
 \begin{align}\label{wuchadingyi}\begin{array}{ll}\hbox{error}:=\\
 10\log_{10}\big[\min\{||\textbf{z}-\textbf{z}_{r}||_{2}/||\textbf{z}||_{2}, ||\textbf{z}+\textbf{z}_{r}||_{2}/||\textbf{z}||_{2}\}\big],\end{array}\end{align}
and  it is reported in dB,  where $\textbf{z}_{r}$ is the recovery result.  Here we say that  a target is  successfully recovered if  the  error is smaller  than $-30$ dB ($\min\{||\textbf{z}-\textbf{z}_{r}||_{2}/||\textbf{z}||_{2}, ||\textbf{z}+\textbf{z}_{r}||_{2}/||\textbf{z}||_{2}\}\leq 0.1\%$).
Real-valued standard  Gaussian measurements are used for
Alternating Minimization.  We found in this simulation that the BlockPR  with the  Fourier-like  measurements performs better than that with random
 measurements. Therefore, we use the Fourier-like measurements for  BlockPR.
 We next   compare the  measurement  cost,   relative recovery  error and   time cost  of   the three  methods.

 \begin{figure}\label{Figure3}
    \flushleft
\includegraphics[width=9.3cm, height=8cm]{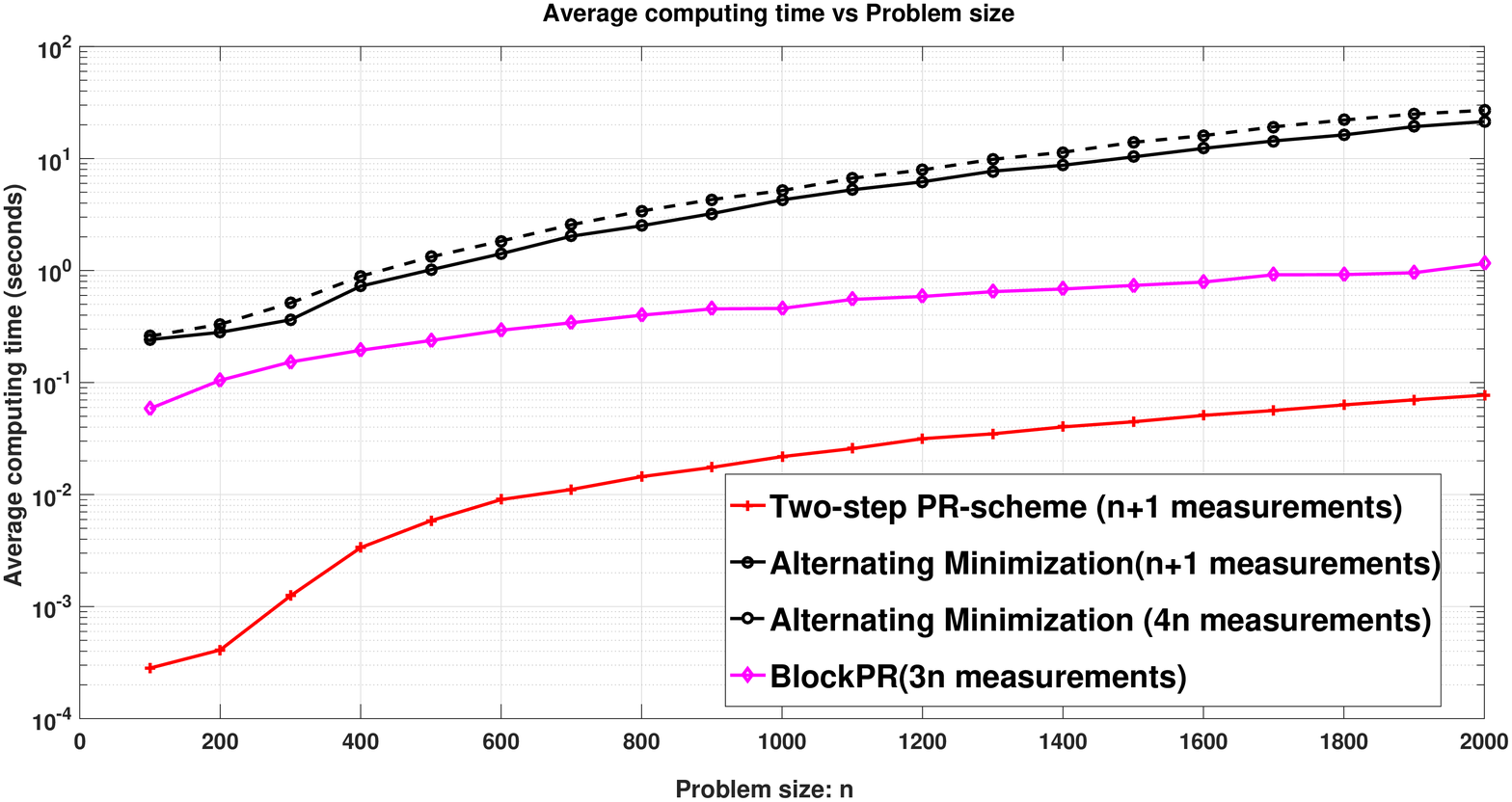}
    \caption{The   computing time   vs problem size $n$  corresponding to two-step PR-scheme, Alternating Minimization and BlockPR.}
\end{figure}

%

 Recall  that the  two-step PR-scheme just  requires  $n+1$ measurements, and
the  BlockPR software  requires at least $3n$   measurements. Therefore for comparing  the computing time
at the same amount of measurements,
we first recover $\textbf{z}$ in \eqref{target} by the  two-step PR-scheme and  Alternating Minimization, respectively.
Both the  two  methods are conducted by $n+1$ measurements for $100$ trials, and their  average  time costs and errors are recorded.
In Fig. IV.2 and Fig. IV.3  we plotted  the numerical results on time cost and errors.
It is observed from the black solid  curve   in  Fig. IV.2  that Alternating Minimization has  the  computational complexity which essentially  scales squarely with the problem size $n$. Actually it follows   from \cite{Netrapalli} that  the theoretic computational complexity  of Alternating Minimization  is $O(n^{2}\log^{2}n(\log n+\log\frac{1}{\epsilon}\log\log\frac{1}{\epsilon}))$, where $\epsilon$ is the computing accuracy. By Proposition  \ref{complex}, the
two-step PR-scheme has the FFT computational complexity $O(n\log n)$ instead.  Obviously, in this simulation   the  two-step PR-scheme costs  much less time than Alternating Minimization.

 The    curve   (in red) in  Fig. IV.3 affirms that,  just requiring $n+1$ measurements, $\textbf{z}$ can be perfectly recovered by the two-step  PR-scheme.
By  direct  observation on the  solid black curve  in   Fig. IV.3, $n+1$ measurements are obviously not sufficient for Alternating Minimization, which is in accordance with \cite{Netrapalli}.
That is, for successfully recovering $\textbf{z}$, more measurements are necessary.

 \begin{figure}\label{Figure3}
    \flushleft
\includegraphics[width=10cm, height=8cm]{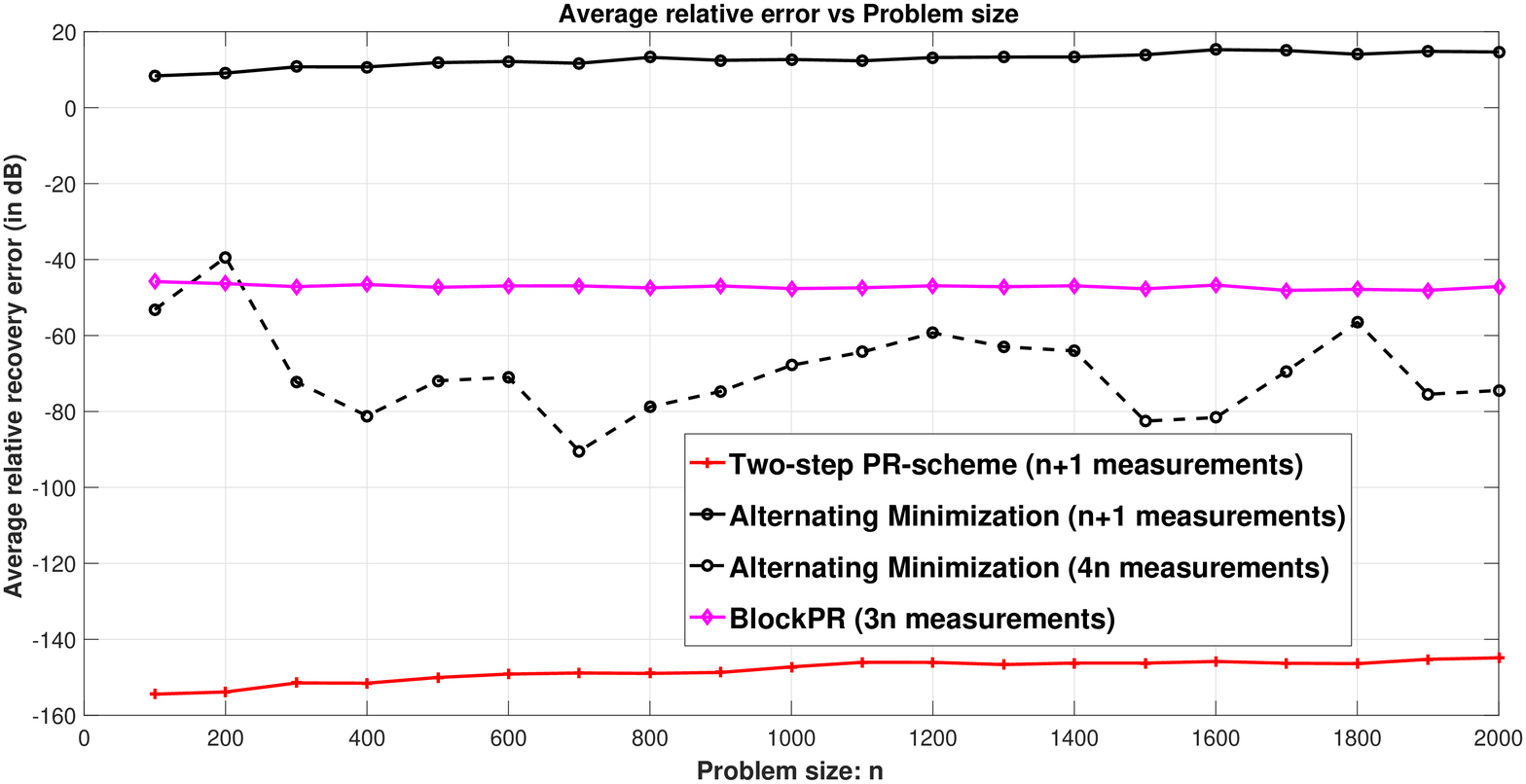}
    \caption{The   recovery error vs problem size $n$  corresponding to two-step PR-scheme, Alternating Minimization and BlockPR.}
\end{figure}

 \begin{figure}\label{Figure4}
  \flushleft
\includegraphics[width=12cm, height=7cm]{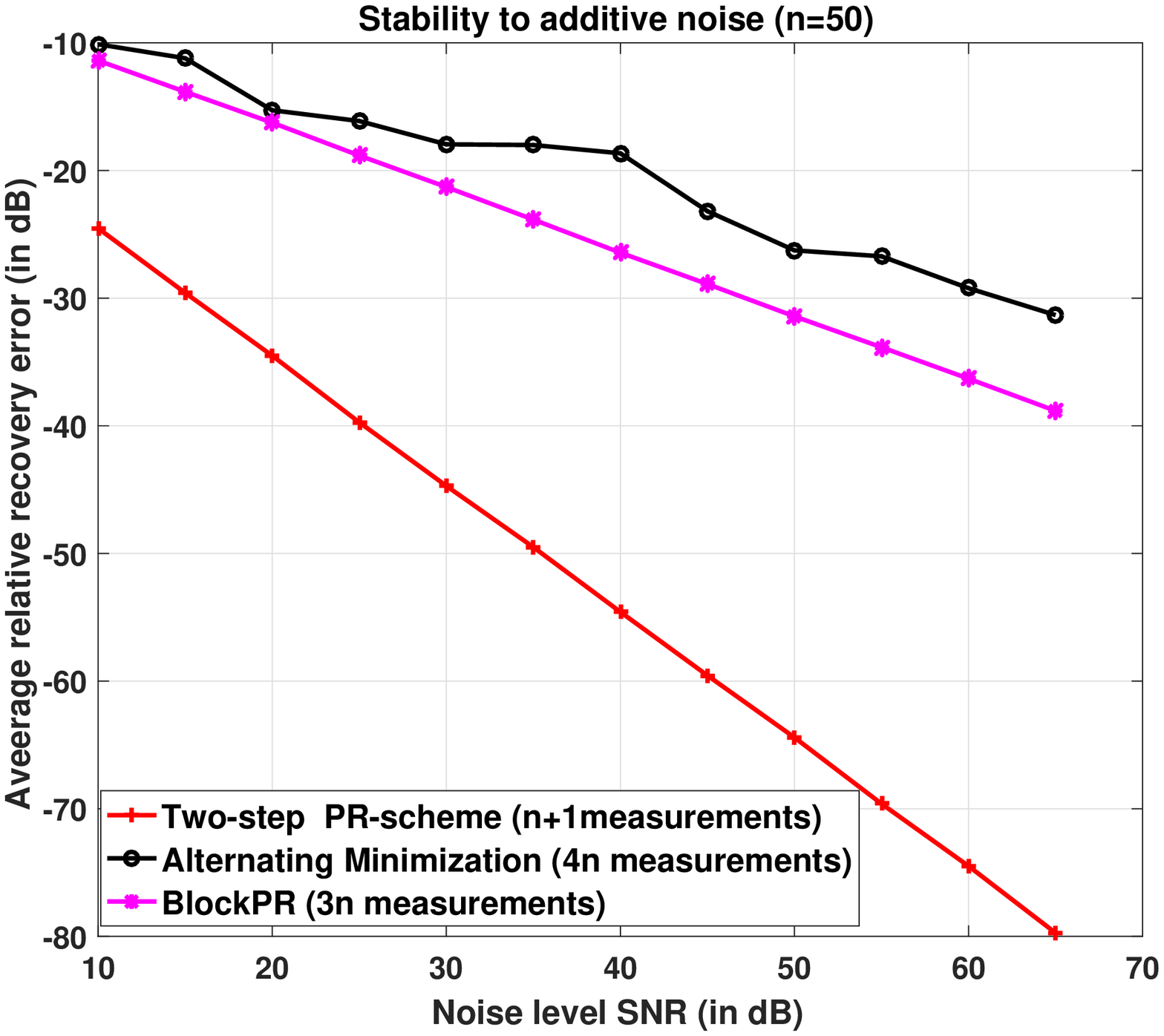}
    \caption{The   recovery error vs the   noise level (SNR) corresponding to the two-step   PR-scheme, Alternating Minimization and BlockPR.}
\end{figure}

 \begin{figure}\label{Figure4}
  \flushleft
\includegraphics[width=13cm, height=7cm]{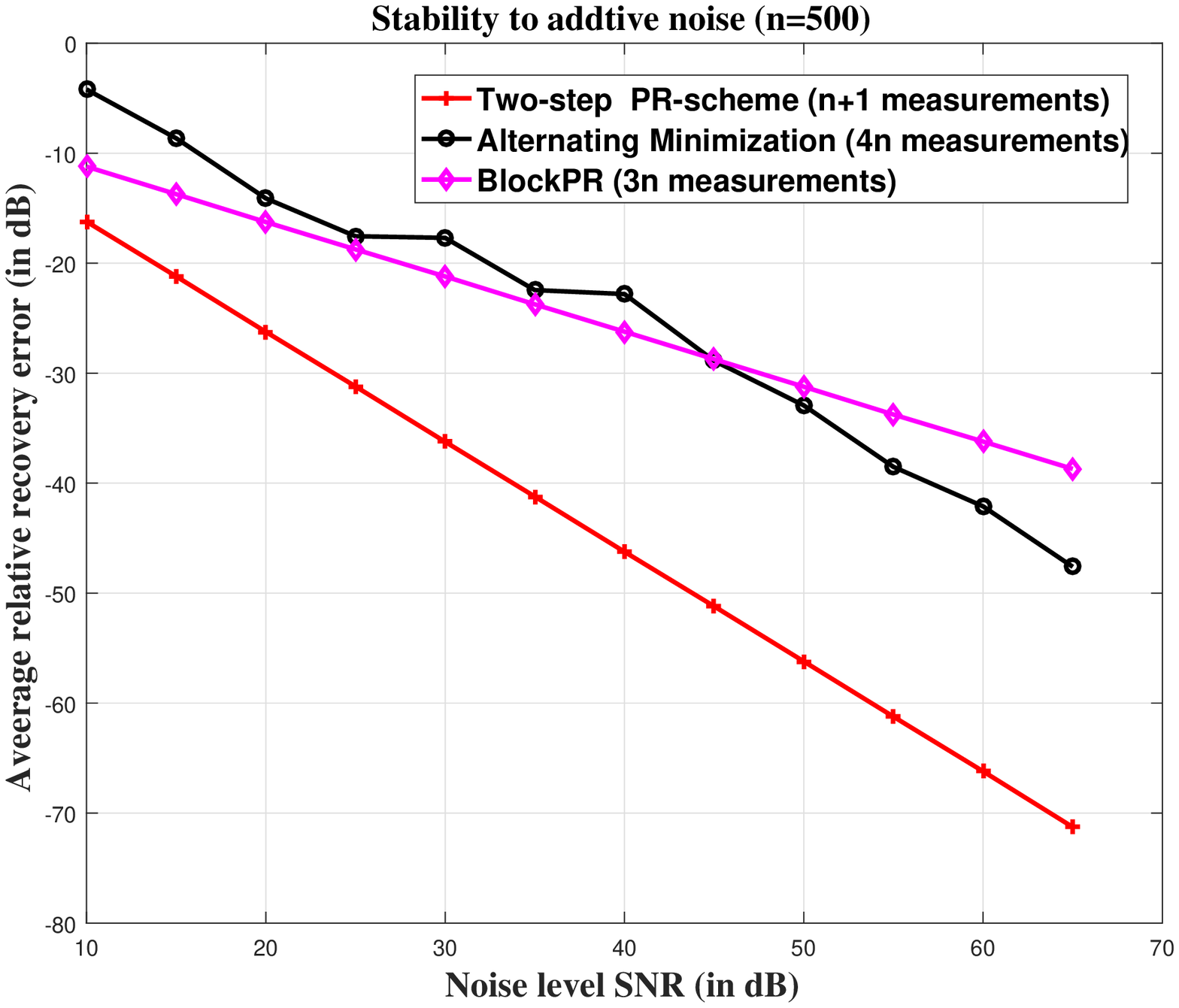}
    \caption{The  recovery error vs the   noise level (SNR) corresponding to the two-step   PR-scheme, Alternating Minimization and BlockPR.}
\end{figure}
Next we continued the simulation for recovering $\textbf{z}$, where $4n$ measurements are used for
 Alternating Minimization, and $3n$ measurements for BlockPR.  We  observed from  Fig. IV.2 that BlockPR  has essentially the FFT computational complexity, but the  two-step PR-scheme has a much smaller constant than  BlockPR. Although
Alternating Minimization    and BlockPR cost  much more measurements and computing time,
it is observed from Fig. IV.3 that the error of the  two-step   PR-scheme is much smaller than theirs. On the other hand, the error  of Alternating Minimization is the second smallest (black and dash curve in Fig. IV.3). Recall that  Alternating Minimization is an iterative method.
Besides on  the amount of random  measurements (the more measurements are used, the better performance it shows with higher probability),   the recovery error  also  depends on  the
 convergence to the target. By \cite{Netrapalli},  Alternating Minimization converges geometrically to the target $\textbf{z}$.
Recall that our     scheme  \eqref{step} is  conducted by Algorithm \ref{separating} and recovery formula  \eqref{fastchogngouxy}.
\begin{figure}\label{Figure4}
  \flushleft
\includegraphics[width=16cm, height=8cm]{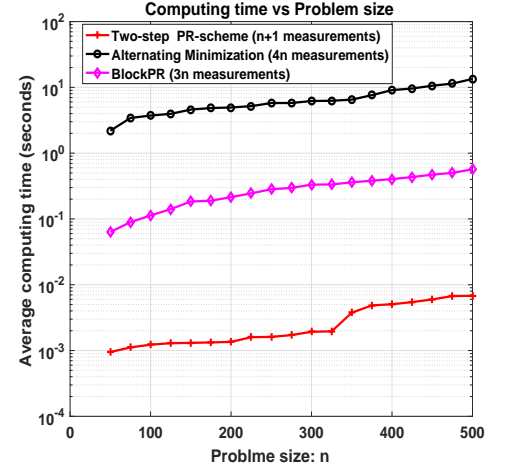}
    \caption{The   computing time  vs the  problem size $n$ corresponding to the two-step   PR-scheme, Alternating Minimization and BlockPR.}
\end{figure}
Neither of the two steps  is iterative instead, and consequently   our scheme is free from the convergence problem in the  computation.
Therefore our recovery error mainly depends on the computing   round-off error. We observed from Fig. IV.3 that  the  round-off error  is very small.

\subsection{Two-step PR-scheme for the random signals in the  noisy setting}
In this subsection we check the  stability   to noise  of our scheme \eqref{step} in the previous simulation, where  any measurement $|\langle \textbf{q}, \textbf{z}\rangle|$  was contaminated by the Gaussian noise
\begin{align}\notag \underline{\hbox{n}}\sim \emph{\textbf{N}}(0,\sigma^{2}).\end{align}
That is, what we observed is
\begin{align}
\widetilde{|\langle \textbf{q}, \textbf{z}\rangle|}=|\langle \textbf{q}, \textbf{z}\rangle|+\underline{\hbox{n}}.
\end{align}

Since the stability  for the recovery formula  \eqref{fastchogngouxy} has been given in Theorem \ref{ghhh},
we  just need to establish the stability for   detection before conducting the numerical simulation  in the noisy setting.
As already  shown in Algorithm \ref{separating} (steps \textbf{5}-\textbf{7}, steps \textbf{11}-\textbf{13}), the  threshold  technique (the threshold value therein  is $0$) was substantially  used in
the detection  strategy in the noiseless setting.
Based on  \eqref{dcg}, for the detection in the noisy setting
we need  to modify the threshold  technique in Algorithm \ref{separating} as follows.
\begin{framed}
 For  $\cup^{2}_{k=1}\textbf{\hbox{cone}}(\mathcal{X}_{k})$ and  a given  threshold value $T$, if the noisy measurement  $\widetilde{|\langle \textbf{g}, \textbf{z}\rangle|} \geq T$, then  the target $\textbf{z}$
is  regarded as not in $\textbf{\hbox{cone}}(\mathcal{X}_{2})$, or else not in  $\textbf{\hbox{cone}}(\mathcal{X}_{1})$.
\end{framed}
A  natural problem is how to choose the  threshold value $T $ such that the target cone can be detected successfully  by  the detection  strategy associated with the above technique.
We give an answer in the following proposition where the lower bound of  $\theta_{1}+\cdots+\theta_{m_{k}}$ is regarded   as the  prior information (The similar  information  was also necessary for the stability of the phase-retrieval in shift-invariant space (Q. Sun et. al \cite{QiyuPR,QiyuPR1})).

\begin{prop}\label{tuilun}
Suppose that  the target $\textbf{z}=\theta_{1}\textbf{x}_{k,1}+\cdots+\theta_{m_{k}}\textbf{x}_{k, m_{k}}\in \textbf{\hbox{cone}}(\mathcal{X}_{k})$ where $k\in \{1, 2\}$, $m_1=2n-1$ and $m_2=n.$ If $\theta_{1}+\cdots+\theta_{m_{k}}\geq r>0$,  then choosing the  threshold value $T:=\frac{r}{2}\min \mathcal{X}^{T}_{1}\textbf{g}$, with at least  the probability $\int^{\frac{r}{2}\min \mathcal{X}^{T}_{1}\textbf{g}}_{-\infty}\frac{1}{\sqrt{2\pi}\sigma}e^{-\frac{(x-\mu)^{2}}{2\sigma^{2}}}dx$,
the target cone can be successfully detected by the modified strategy in the previous  square frame, where $\min \mathcal{X}^{T}_{1}\textbf{g}$ is the minimum
of the $2n-1$ coordinates of the vector $\mathcal{X}^{T}_{1}\textbf{g}$.
\end{prop}
\begin{proof}
The proposition is proved in the Appendix section.
\end{proof}

To check the stability of the scheme in  \eqref{step}, we  conduct the simulation
in   Subsection \ref{noiseless} by  adding the  Gaussian noise to the magnitude measurements.
Following \cite{WY}, the variance is chosen such that the desired signal to noise ratio  (SNR)
is expressed by
$$\begin{array}{lllllllllllllllll}\notag  \hbox{SNR}=10\log_{10}\Big(\frac{||\mathcal{M}^{T}f||^{2}_{2}}{m\sigma^{2}}\Big),\end{array}$$
where $\mathcal{M}$ is the measurement matrix  having $m$  column vectors. SNR is also reported in dB.
For the  scheme  \eqref{step} in the simulation, $\mathcal{M}=[\textbf{g}, \textbf{f}_{1}, \ldots, \textbf{f}_{n}]$.
We conducted  the two-step  PR-scheme, Alternating Minimization and BlockPR on the random signal
$\textbf{z}$ in \eqref{target} for $100$ trials,  where $50\leq n\leq 500$.  
We plotted the  average error to the noise level  in Fig. IV.4 ($n=50$) and Fig. IV.5 ($n=500$), and the average computing time
to the  dimension size $n$ in Fig. IV.6.

It was  observed from Fig. IV.4-5 that for successfully
recovering the  target $\textbf{z}$ (i.e. the error is smaller than $-30$ dB), the requirement on the noise level of the  two-step PR-scheme  is weakest.
As SNR being large sufficiently,
the  two-step PR-scheme has the   smallest error  among the three methods, which coincides  with the  results in noiseless setting as shown in Fig. IV.3.
Moreover, Fig. IV.6 confirms again  that     the  two-step PR-scheme required the  less  computing time.

\section{Appendix}
\subsection{Proof of Theorem \ref{twocone}}\label{proofttwocone}

The proof will be concluded for the cases of $L=2$ and $L>2,$ respectively.

\textbf{\textbf{Case of $L=2$.}} For this case, \eqref{separatingPRgeneral} is equivalent to  that either
\begin{align}\label{separatingPR} \begin{array}{lllllllllllllllll} \hbox{invim}(\mathcal{R}(\mathcal{X}_{1}^{T})\cap \mathbb{R}^{+,m_{1}})\cap \mathcal{N}(X^{T}_{2})\neq\emptyset ,\end{array}\end{align}
or
\begin{align}\label{separatingPRadd}\begin{array}{lllllllllllllllll} 
\hbox{invim}(\mathcal{R}(\mathcal{X}_{2}^{T})\cap \mathbb{R}^{+,m_{2}})\cap \mathcal{N}(X^{T}_{1})\neq\emptyset\end{array}\end{align}
holds.

\textcolor[rgb]{0.00,0.00,1.00}{Sufficiency}:  If, for example, $\hbox{invim}(\mathcal{R}(\mathcal{X}_{1}^{T})\cap \mathbb{R}^{+,m_{1}})\cap \mathcal{N}(X^{T}_{2})\neq\emptyset$, then we can use a measurement vector $\textbf{g}\in \hbox{invim}(\mathcal{R}(\mathcal{X}_{1}^{T})\cap \mathbb{R}^{+,m_{1}})\cap \mathcal{N}(X^{T}_{2})$ as an detector   to complete the detection. Specifically, for any fixed  target nonzero vector $\textbf{z}\in \bigcup^{2}_{k=1} \textbf{\hbox{cone}}(\mathcal{X}_{k})$, if $|\langle \textbf{z}, \textbf{g}\rangle|\neq0$, then $\textbf{z}\notin \textbf{\hbox{cone}}(\mathcal{X}_{2})$ but
$\textbf{z}\in \textbf{\hbox{cone}}(\mathcal{X}_{1})$. If $|\langle \textbf{z}, \textbf{g}\rangle|=0$, then $\textbf{z}\notin \textbf{\hbox{cone}}(\mathcal{X}_{1})$ but
$\textbf{z}\in \textbf{\hbox{cone}}(\mathcal{X}_{2})$.

\textcolor[rgb]{0.00,0.07,1.00}{Necessity}:  Suppose that  we can use an detector  $\textbf{g}\in \mathbb{R}^n$ to  detect  the source of any  $\textbf{z}\in \bigcup^{2}_{k=1} \textbf{\hbox{cone}}(\mathcal{X}_{k})$.
Then \begin{align}\label{identifier} \begin{array}{lllllllllllllllll} \{|\langle \textbf{g}, \textbf{y}\rangle|: \textbf{y}\in \textbf{\hbox{cone}}(\mathcal{X}_{1})\backslash\{0\}\}\\
\cap \{|\langle \textbf{g}, \textbf{y}\rangle|: \textbf{y}\in \textbf{\hbox{cone}}(\mathcal{X}_{2})\backslash\{0\}\}=\emptyset.\end{array}\end{align}
If $\{|\langle \textbf{g}, \textbf{y}\rangle|: \textbf{y}\in \textbf{\hbox{cone}}(\mathcal{X}_{i})\backslash\{0\}\}\cap \mathbb{R}^{+}\neq\emptyset$, then it is straightforward to check  that
$\{|\langle \textbf{g}, \textbf{y}\rangle|: \textbf{y}\in \textbf{\hbox{cone}}(\mathcal{X}_{i})\backslash\{0\}\}\supseteq\mathbb{R}^{+}.$ Therefore \eqref{identifier} is equivalent to the condition that one of the two sets therein  is $\mathbb{R}^{+}$ while the other is $\{0\}$. Without losing generality, we can assume that $\{|\langle \textbf{g}, \textbf{y}\rangle|: \textbf{y}\in \textbf{\hbox{cone}}(\mathcal{X}_{1})\backslash\{0\}\}=\mathbb{R}^{+}$. This implies that $\mathcal{X}^{T}_{1}\textbf{g}\in \mathbb{R}^{+,m_{1}}$.  In fact, if not, then there exist $(\theta_{1},
\ldots, \theta_{m_{1}})\succ0$ such that $\langle \sum^{m_{1}}_{k=1}\theta_{k}\textbf{x}_{1,k}, \textbf{g}\rangle=0$, which leads to a contraction with the assumption.

\textbf{\textbf{Case of $L>2$.}} Invoking  the result of the  case of $L=2$,   the condition in \eqref{separatingPRgeneral} is equivalent to that
  each sub-union
 $\textbf{\hbox{cone}}(\mathcal{X}_{k})\cup\textbf{\hbox{cone}}(\mathcal{X}_{l})$ is detectable.

\textcolor[rgb]{0.00,0.07,1.00}{ Necessity}:  If the UoC $\bigcup^{L}_{k=1} \textbf{\hbox{cone}}(\mathcal{X}_{k})$ is detectable, then by  Definition \ref{detectable},  each sub-union
 $\textbf{\hbox{cone}}(\mathcal{X}_{k})\cup\textbf{\hbox{cone}}(\mathcal{X}_{l})$ is detectable.

\textcolor[rgb]{0.00,0.07,1.00}{ Sufficiency}:
When \eqref{separatingPR} holds for $k=2$, for example, $\hbox{invim}(\mathcal{R}(\mathcal{X}_{1}^{T})\cap \mathbb{R}^{+,m_{1}})\cap \mathcal{N}(X^{T}_{2})\neq\emptyset$.
We pick vector   $\textbf{g}\in\hbox{invim}(\mathcal{R}(\mathcal{X}_{1}^{T})\cap \mathbb{R}^{+,m_{1}})\cap \mathcal{N}(X^{T}_{2})$.
If $|\langle \textbf{g}, \textbf{z}\rangle|>0,$ then $\textbf{z}\notin\textbf{\hbox{cone}}(\mathcal{X}_{2})$. Conversely, if $|\langle \textbf{g}, \textbf{z}\rangle|=0,$ then $\textbf{z}\notin\textbf{\hbox{cone}}(\mathcal{X}_{1})$.   Now there are two
cases: (a) If $\textbf{z}\notin\textbf{\hbox{cone}}(\mathcal{X}_{1})$, then similarly  we next determine whether $ \textbf{z}\notin\textbf{\hbox{cone}}(\mathcal{X}_{2})$ or $ \textbf{z}\notin\textbf{\hbox{cone}}(\mathcal{X}_{3})$.  (b) If $\textbf{z}\notin\textbf{\hbox{cone}}(\mathcal{X}_{2})$,  then we next need to  determine  whether $ \textbf{z}\notin\textbf{\hbox{cone}}(\mathcal{X}_{1})$ or $ \textbf{z}\notin\textbf{\hbox{cone}}(\mathcal{X}_{3})$. The exclusion procedures can go forward due to  \eqref{separatingPR}.
After $L-1$ exclusions, we can detect the target cone where $\textbf{z}$ lies.

\subsection{The proof Theorem \ref{ghhh}}
The measurements for the recovery \eqref{fastchogngouxy} are contaminated by $\underline{\textbf{\hbox{n}}}=[\underline{\hbox{n}}_{1},\ldots, \underline{\hbox{n}}_{\gamma}]$, namely, the measurements we obtained are
\begin{align}\notag\begin{array}{lllll} \{\widetilde{|\langle \textbf{f}_{k}, \textbf{z}\rangle|}\}^{\gamma}_{k=1}=\{|\langle \textbf{f}_{k}, \textbf{z}\rangle|+\underline{\hbox{n}}_{k} \}^{\gamma}_{k=1}.\end{array}
\end{align}
In the procedure of recovering $\mathfrak{P}f$, the emerging error is
\begin{align}\begin{array}{lllll}\label{Errordd}\hbox{Error}=&
\textcolor[rgb]{0.00,0.07,1.00}{\hbox{FFT}}\Big(\hbox{diag}^{-1}(\textcolor[rgb]{0.00,0.07,1.00}{\hbox{FFT}}((\mathfrak{P}\textsf{f}_{1})^{T}))\\
&\times \textcolor[rgb]{0.00,0.07,1.00}{\hbox{IFFT}}\Big(\left[\begin{array}{lllllllllllllllll}
1&0&0&\cdots&0\\
-\delta_{2}&1&0&\cdots&0\\
\vdots&\vdots&\vdots&\ddots&\vdots\\
-\delta_{\gamma}&0&0&\cdots&1
\end{array}\right]\left[\begin{array}{lllllllllllllllll}
\underline{\hbox{n}}_{1}\\
\underline{\hbox{n}}_{2}\\
\vdots\\
\underline{\hbox{n}}_{\gamma}
\end{array}\right]\Big)\Big).
\end{array}
\end{align}
For any vector $\textbf{z}\in \mathbb{R}^{\gamma}$, it is easy to check that  $||\textcolor[rgb]{0.00,0.07,1.00}{\hbox{FFT}}\textbf{z}||_{2}=\frac{1}{\sqrt{\gamma}}||\textbf{z}||_{2}$
and $||\textcolor[rgb]{0.00,0.07,1.00}{\hbox{IFFT}}\textbf{z}||_{2}=\sqrt{\gamma}||\textbf{z}||_{2}$. By this property, Error in \eqref{Errordd}
is estimated as follows,
\begin{align}\label{hhhhh}\begin{array}{lllllllllllllllll} ||\hbox{Error}||_{2}&\leq  \frac{\sqrt{2||\underline{\textbf{n}}||^{2}_2+\max\{\delta_{2}, \ldots, \delta_{\gamma}\}(\gamma-1)\underline{\hbox{n}}^{2}_{1}}}{\min|\textcolor[rgb]{0.00,0.07,1.00}{\hbox{FFT}}(\mathfrak{P}\textsf{f}_{1})|}.
\end{array}\end{align}
We next estimate the probability \begin{align}\begin{array}{lllllllllllllllll} \label{gxxcc} P_{r}(|\frac{\underline{\hbox{n}}^{2}_{1}+\ldots+\underline{\hbox{n}}^{2}_{\gamma}}{\gamma}-\underline{\hbox{n}}^{2}_{1}|>\epsilon)\\
=P_{r}(|\frac{\underline{\hbox{n}}^{2}_{2}+\ldots+\underline{\hbox{n}}^{2}_{\gamma}}{\gamma-1}-\underline{\hbox{n}}^{2}_{1}|>\frac{\gamma}{\gamma-1}\epsilon)\\
\leq P_{r}(|\frac{\underline{\hbox{n}}^{2}_{2}+\ldots+\underline{\hbox{n}}^{2}_{\gamma}}{\gamma-1}-\sigma^2|+|\sigma^2-\underline{\hbox{n}}^{2}_{1}|>\frac{\gamma}{\gamma-1}\epsilon)\\
\leq P_{r}(|\frac{\underline{\hbox{n}}^{2}_{2}+\ldots+\underline{\hbox{n}}^{2}_{\gamma}}{\gamma-1}-\sigma^2|>\frac{\gamma}{2(\gamma-1)}\epsilon)\\
\quad +
P_{r}(|\sigma^2-\underline{\hbox{n}}^{2}_{1}|>\frac{\gamma}{2(\gamma-1)}\epsilon)\\
=1-\Phi_{\gamma-1}(\gamma-1+\frac{\gamma \epsilon}{2\sigma^{2}})+\Phi_{\gamma-1}(\gamma-1-\frac{\gamma \epsilon}{2\sigma^{2}})\\
\quad +1-\Phi_{1}(1+\frac{\gamma}{2(\gamma-1)\sigma^2}\epsilon)+ \Phi_{1}(1-\frac{\gamma}{2(\gamma-1)\sigma^2}\epsilon).\end{array}\end{align}
Therefore with the probability at least \begin{align}\notag \begin{array}{lllllllllllllllll}-1+\Phi_{\gamma-1}(\gamma-1+\frac{\gamma \epsilon}{2\sigma^{2}})+\Phi_{1}(1+\frac{\gamma}{2(\gamma-1)\sigma^2}\epsilon)\\
-\Phi_{\gamma-1}(\gamma-1-\frac{\gamma \epsilon}{2\sigma^{2}})-\Phi_{1}(1-\frac{\gamma}{2(\gamma-1)\sigma^2}\epsilon),\end{array}\end{align}
it holds that $P_{r}(|\frac{\underline{\hbox{n}}^{2}_{1}+\ldots+\underline{\hbox{n}}^{2}_{\gamma}}{\gamma}-\underline{\hbox{n}}^{2}_{1}|\leq\epsilon)$.
By \eqref{hhhhh} and \eqref{gxxcc}, with at least the above probability, it holds that
\begin{align}\label{hhhnew}\begin{array}{lllllllllllllllll} ||\hbox{Error}||_{2}&\leq  \frac{\sqrt{2||\underline{\textbf{n}}||^{2}_2+\max\{\delta_{2}, \ldots, \delta_{\gamma}\}[(\gamma-1)\epsilon+\frac{\gamma-1}{\gamma}||\underline{\textbf{n}}||^{2}_2]}}{\min|\textcolor[rgb]{0.00,0.07,1.00}{\hbox{FFT}}(\mathfrak{P}\textsf{f}_{1})|}.
\end{array}\end{align}

\subsection{The proof Theorem \ref{tuilun}}
Suppose that  the following event  \begin{align}\label{fffff} \begin{array}{lllllllllllllllll}\underline{\hbox{n}}< \frac{r}{2}\min \mathcal{X}^{T}_{1}\textbf{g}\end{array}\end{align}
holds.
If $\textbf{z}\in \textbf{\hbox{cone}}(\mathcal{X}_{1})$, then
\begin{align}\begin{array}{lllllllllllllllll}\label{kkks1}\widetilde{|\langle \textbf{g}, \textbf{z}\rangle|}\\
=|\langle \textbf{g}, \textbf{z}\rangle|+\underline{\hbox{n}}\geq (\theta_{1}+\ldots+\theta_{m_{k}})\min \mathcal{X}^{T}_{1}\textbf{g}-
|\underline{\hbox{n}}|\\
\geq \frac{r}{2}\min \mathcal{X}^{T}_{1}\textbf{g}.\end{array}\end{align}
On the other hand, if  $\textbf{z}\in \textbf{\hbox{cone}}(\mathcal{X}_{2})$,  then
\begin{align}\label{kkks2}\begin{array}{lllllllllllllllll} \widetilde{|\langle \textbf{g}, \textbf{z}\rangle|}=\underline{\hbox{n}}< \frac{r}{2}\min \mathcal{X}^{T}_{1}\textbf{g}.\end{array}\end{align}
It follows from \eqref{kkks1} and \eqref{kkks2} that the target cone can be successfully detected by the modified detection  strategy.
On the other hand, \eqref{fffff} holds with the probability $\int^{\frac{r}{2}\min \mathcal{X}^{T}_{1}\textbf{g}}_{-\infty}\frac{1}{\sqrt{2\pi}\sigma}e^{-\frac{(x-\mu)^{2}}{2\sigma^{2}}}dx$. The proof is concluded.

\tabcolsep 30pt

\vspace*{7pt}

\end{document}